\documentclass[10pt,reqno,a4paper]{amsart}
\usepackage[centering]{geometry}
\usepackage[british]{babel}
\usepackage{amssymb,amstext,amsthm,eucal,bbm,mathrsfs,amscd,bbold,pifont, tensor}
\usepackage{array}
\usepackage[svgnames,table]{xcolor}
\usepackage[T1]{fontenc}
\usepackage[utf8]{inputenc}
\usepackage{enumitem}
\usepackage[small]{eulervm}
\usepackage{microtype}
\usepackage{booktabs,caption}
\usepackage{verbatim}
\usepackage{braket}
\usepackage{apptools}
\AtAppendix{\counterwithin{lemma}{section}}
\usepackage{mdframed}
\usepackage[nocompress]{cite}
\usepackage{tikz,pgfplots}
\usetikzlibrary{calc,arrows,cd}
\usepackage[unicode,pagebackref=true]{hyperref}
\hypersetup{%
  pdftitle   = {The gauging procedure and carrollian gravity},
  pdfkeywords = {Lie algebras, kinematical, spacetime, homogeneous spaces, Poincaré, gauging, Cartan},
  pdfauthor  = {José Figueroa-O'Farrill, Emil Have, Stefan Prohazka, Jakob Salzer},
  pdfcreator = {\LaTeX\ with package \flqq hyperref\frqq},
  linkcolor=NavyBlue,
  citecolor=ForestGreen,
  urlcolor=OrangeRed,
  anchorcolor=OrangeRed,
  colorlinks=true}
\renewcommand*{\backref}[1]{}
\renewcommand*{\backrefalt}[4]{%
  \ifcase #1%
  \or [Page~#2.]%
  \else [Pages~#2.]%
  \fi%
}
\theoremstyle{plain}
\newtheorem{lemma}{Lemma}
\newtheorem{proposition}[lemma]{Proposition}

\theoremstyle{definition}

\newcommand{\g}{\mathfrak{g}}
\newcommand{\h}{\mathfrak{h}}

\renewcommand{\c}{\mathfrak{c}}
\renewcommand{\d}{\partial}

\newcommand{\gl}{\mathfrak{gl}}

\renewcommand{\r}{\mathfrak{r}}

\newcommand{\iso}{\mathfrak{iso}}
\newcommand{\so}{\mathfrak{so}}

\newcommand{\m}{\mathfrak{m}}

\newcommand{\eX}{\mathscr{X}}

\ifdefined\C
\renewcommand{\C}{\boldsymbol{C}}
\else
\newcommand{\C}{\boldsymbol{C}}
\fi

\newcommand{\B}{\boldsymbol{B}}

\renewcommand{\P}{\boldsymbol{P}}

\renewcommand{\L}{\boldsymbol{L}}

\newcommand{\Tr}{\operatorname{Tr}}

\newcommand{\Ad}{\operatorname{Ad}}
\newcommand{\ad}{\operatorname{ad}}

\newcommand{\id}{\operatorname{id}}

\newcommand{\Pbar}{\overline{P}}
\newcommand{\Hbar}{\overline{H}}

\newcommand{\1}{\mathbb{1}}

\newcommand{\RR}{\mathbb{R}}

\newcommand{\ZZ}{\mathbb{Z}}

\newcommand{\vol}{\operatorname{dvol}}

\newcommand{\GL}{\operatorname{GL}}

\newcommand{\SO}{\operatorname{SO}}

\newcommand{\eL}{\mathscr{L}}

\newcommand{\zLC}{\mathsf{LC}}
\newcommand{\zAdSC}{\mathsf{AdSC}}
\newcommand{\zAdSCp}{\mathsf{(A)dSC}}
\newcommand{\zdSC}{\mathsf{dSC}}

\newcommand{\zSpi}{\mathsf{Spi}}

\newcommand{\zTi}{\mathsf{Ti}}
\newcommand{\zAdS}{\mathsf{AdS}}
\newcommand{\zC}{\mathsf{C}}

\newcommand{\pd}{\partial}
\newcommand{\D}{\partial}

\newcommand{\cm}{\checkmark}
\definecolor{gris}{rgb}{0.5,0.5,0.5}
\definecolor{darkgreen}{rgb}{0.0,0.5,0.0}

\allowdisplaybreaks[0]
\numberwithin{equation}{section}
\usetikzlibrary{shapes.misc,snakes}
\tikzset{cross/.style={cross out, draw=black, thick, minimum size=2*(#1-\pgflinewidth), inner sep=0pt, outer sep=0pt},
cross/.default={3pt}}
\pgfdeclarelayer{nodelayer}
\pgfdeclarelayer{edgelayer}
\pgfsetlayers{edgelayer,nodelayer,main}
\tikzstyle{ghost}=[fill=none, draw=none, shape=circle]
\tikzstyle{cross1}=[fill=none, draw=black, shape=circle]
\tikzstyle{grey line}=[-, draw={rgb,255: red,128; green,128; blue,128}]
\tikzstyle{blue line}=[-, draw=blue, fill={rgb,255: red,162; green,246; blue,255}]
\tikzstyle{blue line2}=[-, draw=blue]
\tikzstyle{dash blue}=[-, draw=blue, dashed]
\tikzstyle{dash red}=[-, draw=red, dashed, thick]
\tikzstyle{dash black}=[-, draw=black, dashed, thick]
\tikzstyle{dash grey}=[-, draw={rgb,255: red,128; green,128; blue,128}, dashed]
\tikzstyle{thick black line}=[-, thick]
\tikzstyle{blue fill}=[-, draw=none, fill={rgb,255: red,126; green,214; blue,255}]
\tikzstyle{black line}=[-]
\tikzstyle{thick red}=[-,draw=red,thick]
\tikzstyle{red fill}=[-, fill={rgb,255: red,255; green,162; blue,164}, draw={rgb,255: red,255; green,0; blue,4}]
\tikzstyle{purple fill}=[-, fill={rgb,255: red,128; green,0; blue,255}, draw={rgb,255: red,100; green,15; blue,128}]

\begin{document}

\title{The gauging procedure and carrollian gravity}
\author[Figueroa-O'Farrill]{José Figueroa-O'Farrill}
\author[Have]{Emil Have}
\author[Prohazka]{Stefan Prohazka}
\author[Salzer]{Jakob Salzer}
\address[JMF,EH,SP]{Maxwell Institute and School of Mathematics, The University
  of Edinburgh, James Clerk Maxwell Building, Peter Guthrie Tait Road,
  Edinburgh EH9 3FD, Scotland, United Kingdom}
\address[JS]{Physique Théorique et Mathématique, Université libre de Bruxelles and
International Solvay Institutes, Campus Plaine C.P. 231, B-1050
Bruxelles, Belgium
}
\email[JMF]{\href{mailto:j.m.figueroa@ed.ac.uk}{j.m.figueroa@ed.ac.uk}, ORCID: \href{https://orcid.org/0000-0002-9308-9360}{0000-0002-9308-9360}}
\email[EH]{\href{mailto:emil.have@ed.ac.uk}{emil.have@ed.ac.uk}, ORCID: \href{https://orcid.org/0000-0001-8695-3838}{0000-0001-8695-3838}}
\email[SP]{\href{mailto:stefan.prohazka@ed.ac.uk}{stefan.prohazka@ed.ac.uk}, ORCID: \href{https://orcid.org/0000-0002-3925-3983}{0000-0002-3925-3983}}
\email[JS]{\href{mailto:jakob.salzer@ulb.be}{jakob.salzer@ulb.be}, ORCID: \href{https://orcid.org/0000-0002-9560-344X}{0000-0002-9560-344X}}
\begin{abstract}
  We discuss a gauging procedure that allows us to construct
  lagrangians that dictate the dynamics of an underlying Cartan
  geometry. In a sense to be made precise in the paper, the starting
  datum in the gauging procedure is a Klein pair corresponding to a
  homogeneous space. What the gauging procedure amounts to is the
  construction of a Cartan geometry modelled on that Klein geometry,
  with the gauge field defining a Cartan connection. The lagrangian
  itself consists of all gauge-invariant top-forms constructed from
  the Cartan connection and its curvature. After demonstrating that
  this procedure produces four-dimensional General Relativity upon
  gauging Minkowski spacetime, we proceed to gauge all
  four-dimensional maximally symmetric carrollian spaces: Carroll,
  (anti-)de Sitter--Carroll and the lightcone. For the first three of
  these spaces, our lagrangians generalise earlier first-order
  lagrangians. The resulting theories of carrollian gravity all take
  the same form, which seems to be a manifestation of model mutation
  at the level of the lagrangians. The odd one out, the lightcone, is
  not reductive and this means that although the equations of motion take
  the same form as in the other cases, the geometric interpretation is
  different. For all carrollian theories of gravity we obtain
  analogues of the Gauss--Bonnet, Pontryagin and Nieh--Yan topological
  terms, as well as two additional terms that are intrinsically
  carrollian and seem to have no lorentzian counterpart. Since we
  gauge the theories from scratch this work also provides a no-go
  result for the electric carrollian theory in a first-order
  formulation.
\end{abstract}
\thanks{EMPG-22-10}
\maketitle
\tableofcontents

\section{Introduction  and summary}
\label{sec:introduction}

\subsection{Introduction}
\label{sec:introduction-sub}

The scope of this work is two-fold: we first introduce a systematic
gauging procedure rooted in Cartan geometry, which we then use to
construct carrollian theories of gravity for all maximally symmetric
carrollian spacetimes.

It is a frequent claim in the literature that one obtains General
Relativity by ``gauging the Poincaré algebra''. Taking this at face
value one would be forgiven for thinking that the Einstein equations
are somehow derivable from the Poincaré algebra without any further
choices. As we will discuss in detail in this paper, one needs in fact
as an additional ingredient a preferred subalgebra, in this case the
Lorentz algebra. These two ingredients
$(\g,\h)=(\textrm{Poincaré},\textrm{Lorentz})$ form a so-called Klein
pair which specifies a particular homogeneous space of the Poincaré
group, namely Minkowski spacetime.  What is usually known as ``gauging
the Poincaré algebra'' is then just the construction of a Cartan
geometry modelled on Minkowski spacetime.\footnote{In order to have an
  intuitive picture, one can consider Cartan geometry as a
  generalisation of riemannian geometry. While the latter describes
  curved versions of euclidean space, a Cartan geometry describes
  curved versions of an arbitrary homogeneous space. We will be more
  precise in the main text; consult \cite{Wise:2006sm} for a
  delightful introduction to Cartan geometry and \cite{MR1453120} for
  a good introductory textbook.} Of course, Minkowski spacetime is not
the only homogeneous spacetime of the Poincaré group. For example, the
carrollian limit of $\zAdS$, known variously as
anti-de~Sitter--Carroll ($\zAdSC$) \cite{Figueroa-OFarrill:2018ilb}
or, as explained in \cite{Figueroa-OFarrill:2021sxz}, also as the
blow-up ($\zTi^\pm$) of either future or past timelike infinity in
Minkowski spacetime, is also a homogenous space of the Poincaré
group. As we shall see in this paper, ``gauging the Poincaré algebra''
in this case does not lead to General Relativity, but in fact to a
version of carrollian gravity.

An essential ingredient of a Cartan geometry on a manifold $M$
modelled on a homogeneous space $G/H$ is a Cartan connection $A$, a
Lie algebra valued $1$-form on a principal $H$-bundle over $M$.  If
the Klein pair $(\g,\h)$ is reductive, so that $\g = \h \oplus \m$
with $[\h,\m] \subset \m$, the Cartan connection splits as
\begin{align}
  A = \theta + \omega,
\end{align}
with $\theta$, the $\m$-component, a soldering form and $\omega$, the
$\h$-component, an Ehresmann connection. Even if the Klein pair is not
reductive, we may define the soldering form as the projection of $A$
to $\g/\h$, but in the non-reductive case no part of $A$ can be
interpreted as an Ehresmann connection. The Klein pair for Minkowski
spacetime is reductive, so that we may split
$A = \theta^m P_m + \tfrac12 \omega^{mn} L_{mn}$, where $P_{m}$ are
the generators of spacetime translations and $L_{mn}$ generate Lorentz
transformations. The curvature of $A$ is given by
$F = dA + \tfrac12[A,A]$ and it measures how the Cartan geometry
deviates from the homogeneous model. Indeed, a Cartan connection is
flat (i.e., $F=0$) if and only if it is locally isomorphic to the
homogeneous model.

In order to specify a Cartan geometry it would in principle be enough
to impose conditions on the curvature $F$.  One systematic way to
achieve this is for these conditions to arise as the Euler--Lagrange
equations of some variational principle.  To this end one constructs
a gauge-invariant lagrangian top-form built out of wedge products of
$\theta$ and $F$.  For Minkowski spacetime, this well-known
construction leads to the lagrangian
\begin{subequations}
  \label{eq:mink-intro}
\begin{align}
  \label{eq:genform}
    \eL &= \eL_{\mathrm{HP}}  + \eL_{\mathrm{Holst}} + \eL_{\Lambda} \\
   &= \tfrac14 \epsilon_{mnpq} \theta^m \wedge \theta^n \wedge  \label{eq:relativistic-lagrangian-intro}
  \Omega^{pq} + \tfrac\beta2 \theta^m \wedge \theta^n \wedge \Omega_{mn}
  - \tfrac{\Lambda}{4!} \epsilon_{mnpq} \theta^m \wedge 
  \theta^n \wedge \theta^p \wedge \theta^q,
\end{align}
\end{subequations}
where $\Omega^{mn} = d\omega^{mn} +\omega^m{}_p \wedge \omega^{pn}$.
The first term is the Hilbert--Palatini lagrangian, the second term is
the modification due to Holst \cite{Holst:1995pc} and the third term
is a cosmological constant term, both with undetermined relative
coefficients. This procedure also produces the Pontryagin,
Gauss--Bonnet and Nieh--Yan boundary terms (see
Table~\ref{tab:4-forms-mink-sum} for the explicit expressions).
Setting the Holst term to zero, this is of course nothing but the
first-order formulation of General Relativity.

As mentioned above, one clue as to why we must stress the importance of
choosing the Klein pair $(\g,\h)=(\textrm{Poincaré},\textrm{Lorentz})$ of the
Poincaré algebra to reproduce \eqref{eq:mink-intro} is the existence
of four-dimensional homogeneous spaces of the Poincaré group
other than Minkowski spacetime.  One such space, $\zAdSC$, specified
by the Klein pair $(\g,\h)=(\textrm{Poincaré},\iso(3))$ is not
lorentzian, but rather carrollian.\footnote{Indeed, if the action of
  the Poincaré group is effective, no such spacetime could be
  lorentzian, since purely on dimensional grounds any
  Poincaré-invariant metric on a four-dimensional lorentzian manifold must be locally isometric to
  the Minkowski metric.}

Carroll symmetry~\cite{Levy1965,SenGupta1966OnAA} arises from Lorentz
symmetry in the limit where the speed of light goes to zero,
as depicted in Figure~\ref{fig:carroll-lightcones} in terms of what
this limit does to the lightcone.  Carrollian physics can thus be
regarded as describing an ultrarelativistic (or, perhaps more
appropriately, an ultra-local) limit.

\begin{figure}[h!] \centering
  \includegraphics[width=0.9\textwidth]{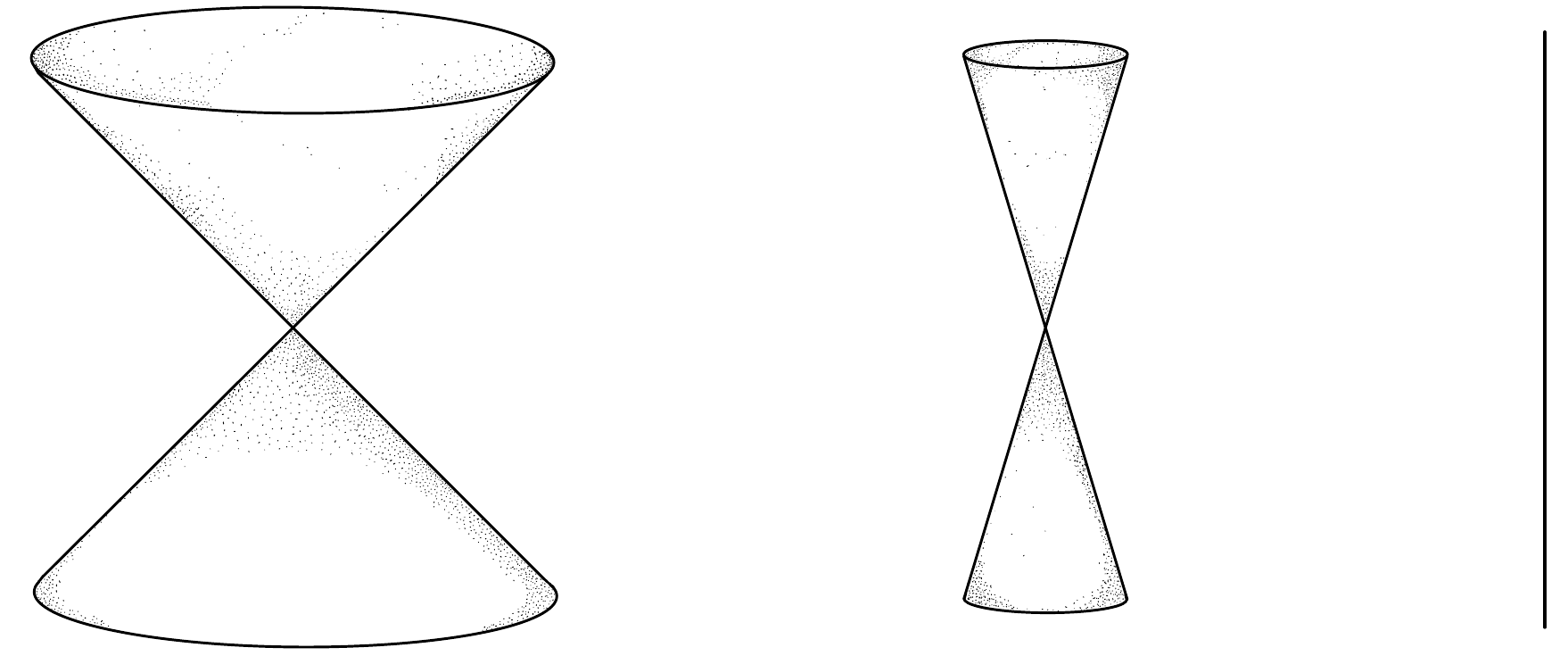}
  \begin{tikzpicture}[overlay] \begin{pgfonlayer}{nodelayer} \node
      [style=ghost] (0) at (-10.7, -0.2) {$c = 1$}; \node
      [style=ghost] (1) at (-4.5, -0.2) {$c\ll 1$}; \node
      [style=ghost] (2) at (-0.3, -0.2) {$c=0$}; \end{pgfonlayer}
  \end{tikzpicture} \vspace{0.4cm} \caption{In the Carroll limit, the
    lightcones close up. In the second scenario, where $c\ll 1$, the
    dynamics can be described by an expansion of General Relativity
    around $c= 0$~\cite{Hansen:2021fxi}. When the lightcone closes
    completely, all motion ceases to exist (although the limiting
    behaviour of tachyons allows for motion~\cite{deBoer:2021jej}),
    and the spacetime becomes ultra-local.}
    \label{fig:carroll-lightcones}
\end{figure}
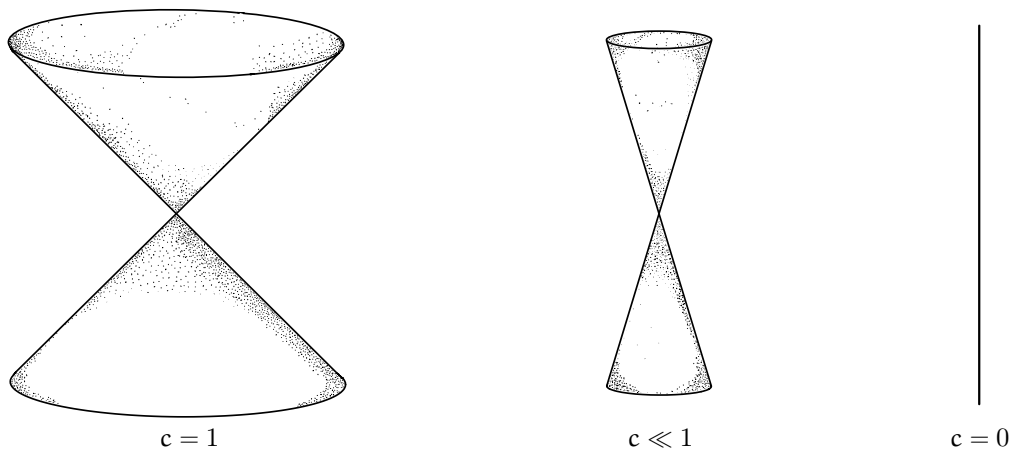

Reasons to be interested in carrollian physics are plentiful, the most
obvious one being its connection to null hypersurfaces in lorentzian
spacetimes. Thus it plays a prominent r\^ole in the discussion of null
infinity~\cite{Duval:2014uva} (which in turn is connected to the
celestial
sphere~\cite{Figueroa-OFarrill:2021sxz,Donnay:2022aba,Bagchi:2022emh}),
black hole horizons~\cite{Donnay:2019jiz} or spatial and timelike
infinity~\cite{Gibbons:2019zfs,Figueroa-OFarrill:2021sxz}.
Consequently, it has been argued to provide a natural starting point
for the discussion of flat space holography, e.g., as discussed
in~\cite{Bagchi:2016bcd} or using carrollian fluids
in~\cite{Ciambelli:2018wre}. Carroll symmetry furthermore appears in
the context of gravitational waves~\cite{Duval:2017els},
cosmology~\cite{deBoer:2021jej}, and fractons~\cite{Bidussi:2021nmp}.
Moreover, the carrollian limit of gravity was considered in
\cite{Henneaux:1979vn} as a way to describe the strong coupling limit
of gravity in the vicinity of a space-like singularity; for other
works on carrollian theories of gravity
see~\cite{Hartong:2015xda, Bergshoeff:2017btm, Guerrieri:2021cdz,
  Henneaux:2021yzg, Hansen:2021fxi}.

Similar to the gauging procedure for the Klein pair of Minkowski
spacetime outlined above, we will construct in this work carrollian
gravitational theories by gauging all four spatially isotropic
maximally symmetric carrollian spacetimes.  These have been
classified\footnote{This generalises the seminal work of Bacry and
  Lévy-Leblond~\cite{Bacry:1968zf} by the addition of the lightcone,
  which was made possible by dropping one of the restrictions
  of~\cite{Bacry:1968zf}.}  in~\cite{Figueroa-OFarrill:2018ilb} and
are given by the following spacetimes: Carroll $\zC$, de
Sitter--Carroll~$\zdSC$, anti-de Sitter--Carroll $\zAdSC$ and the
lightcone $\zLC$, whose Klein pairs are described in
Section~\ref{sec:klein-geometries} and listed in
Table~\ref{tab:klein-pairs}.

Some of the novel features of our approach are:
\begin{enumerate}
\item the way we arrive at the lagrangians is algorithmic;
\item we cover all maximally symmetric carrollian spacetimes;
\item we construct lagrangians from scratch, rather than obtaining
  them from a limit; and
\item in particular, this means we provide all possible invariant
  terms (under the conditions we provide) for the lagrangian, including
  topological terms.
\end{enumerate}
We will conclude this introduction by a brief overview of our results.

\subsection{Summary}

The first lagrangian we obtain is similar in structure to the one of
Minkowski spacetime~\eqref{eq:mink-intro}. It is given by a
three-parameter family with carrollian Hilbert--Palatini, carrollian
Holst and carrollian cosmological constant terms
\begin{align}
\eL_{\zAdSCp} = \eL_{\mathrm{cHP}}  + \beta \eL_{\mathrm{cHolst}} + \Lambda \eL_{\mathrm{c}\Lambda} ,
\end{align}
which are given explicitly in Table~\ref{tab:4-forms-car-sum}. Note
that this lagrangian depends on a parameter $\sigma$ that can be tuned
to give either Carroll, AdS Carroll or dS Carroll. That the equations
of motion are the same in all three cases is a manifestation of the
fact that these spaces are \textit{model mutants}: Klein geometries
giving rise to the same Cartan geometry.  We will discuss this in
Section~\ref{sec:gauging} and at the end of Section~\ref{sec:gauging-poincare}.
\begin{table}[h!]
  \centering
    \rowcolors{2}{blue!10}{white}
    \renewcommand{\arraystretch}{1.3}
  \begin{tabular}{>{$}l<{$}|>{$}c<{$}|>{$}l<{$}}\toprule
    \multicolumn{1}{c|}{$4$-form}                                                                                                           & \multicolumn{1}{c|}{$\partial$} & \multicolumn{1}{c}{carrollian analogue of} \\
    \midrule
  \eL_{\mathrm{cHP}}=  \tfrac12 \epsilon_{abc} \left(\theta^a \wedge \theta^b \wedge \Psi^c + \xi \wedge \theta^a \wedge \Omega^{bc}\right) &                         & \text{Hilbert--Palatini}       \\
  \eL_{\mathrm{cHolst}}=  \tfrac12 \theta^a \wedge \theta^b \wedge \Omega_{ab}                                                              &                         & \text{Holst}                   \\    
  \eL_{c\Lambda}=  \tfrac16 \epsilon_{abc} \theta^a \wedge \theta^b \wedge \theta^c \wedge \xi                                              &                         & \text{Cosmological}            \\
    \tfrac12 \epsilon_{abc}\Psi^a \wedge \Omega^{bc} - \sigma (\eL_{\mathrm{cHP}} + \frac13 \eL_{\textrm{c}\Lambda})                                 & \cm                     & \text{Pontryagin}              \\
    \tfrac14 \Omega^{ab} \wedge \Omega_{ab} - \tfrac12 \sigma \Theta^a \wedge\Theta_a                                                       & \cm                     & \text{Gauss--Bonnet}           \\
        \eL_{\mathrm{cHolst}} - \tfrac12  \Theta^a \wedge\Theta_a = -\tfrac12 d(\theta^a \wedge \Theta_a)                                                                              & \cm                     & \text{Nieh--Yan}               \\
    \tfrac12 \epsilon_{abc} \theta^a \wedge \theta^b \wedge \Theta^c = \tfrac16d(\epsilon_{abc}\theta^a\wedge \theta^b\wedge \theta^c)                                                                        & \cm                     & \text{intrinsic}               \\    
    \tfrac12 \epsilon_{abc} \Theta^a \wedge \Omega^{bc} = \tfrac12 d(\epsilon_{abc}\theta^a\wedge \Omega^{bc}) - \sigma \epsilon_{abc}\theta^a \wedge \theta^b\wedge \Theta^c                                                                                     & \cm                     & \text{intrinsic}               \\
    \bottomrule
  \end{tabular}
  \caption{Summary of carrollian gauge-invariant $4$-forms and
    boundary terms ($\partial$) for $\zAdSC$ ($\sigma = +1$), $\zdSC$
    ($\sigma = -1$) and $\zC$ ($\sigma = 0$). The carrollian
    curvatures are given by
    $\Omega^{ab} = d\omega^{ab} + \omega^a{}_c \wedge \omega^{cb} +
    \sigma \theta^a \wedge \theta^b$ and
    $\Psi^a = d\psi^a + \omega^a{}_b \wedge \psi^b + \sigma \xi \wedge
    \theta^a$ and part of the carrollian torsion is given by
    $\Theta^a = d\theta^a + \omega^a{}_b \wedge \theta^b$.
     }
  \label{tab:4-forms-car-sum}
\end{table}

This lagrangian generalises earlier first-order
lagrangian~\cite{Bergshoeff:2017btm,Guerrieri:2021cdz} in the following
way:
\begin{enumerate}
\item This gauging generalises to nonzero cosmological constant. For
  $\zAdSC$ we set $\sigma = +1$, while for $\zdSC$ we take
  $\sigma = -1$. This lagrangian also has a well-defined flat limit
  $\sigma=0$ upon which we arrive at the lagrangian for flat
  carrollian space. In each case we have the choice to set
  $\eL_{\mathrm{c}\Lambda}=0$. In all cases this provides the most general
  lagrangian with nontrivial dynamics, under the conditions we provide.
\item We have introduced a carrollian Holst term
  $\eL_{\mathrm{cHolst}}$ which is valid for any cosmological
  constant. It shares the characteristic feature of its lorentzian
  counterpart in that it has a nontrivial effect on the equations of
  motion, but in such a way that the solution space is untouched.
\item We have derived novel topological carrollian boundary terms.
  They are analogous to the Pontryagin, Gauss--Bonnet and Nieh--Yan
  terms, with two additional terms that appear to be intrinsically
  carrollian, as listed in Table~\ref{tab:4-forms-car-sum}.
\end{enumerate}
We remark yet again that AdS Carroll, like Minkowski, is a homogeneous
space of the Poincaré group.  This explicitly shows that the sentence
``gauging the Poincaré algebra'' is ambiguous and can lead to
inequivalent lagrangians.  In fact, it is not just ambiguous but
indeed misleading.  The phenomenon of model mutation (to be discussed
below) shows that it is not the Lie algebra structure of $\g$ that is
important, but the structure of $\g$ as a representation of $\h$, at
least in the reductive case.

The remaining maximally symmetric carrollian space is the lightcone,
for which we obtain the lagrangian
\begin{equation*}
  \eL_{\zLC} = \eL_{\mathrm{LcHP}}  + \beta \eL_{\mathrm{LcHolst}} + \Lambda \eL_{\mathrm{Lc}\Lambda},
\end{equation*}
with the explicit expressions for the terms given in Table~\ref{tab:4-forms-LC-sum}.
\begin{table}[h!]
  \centering
    \rowcolors{2}{blue!10}{white}
    \renewcommand{\arraystretch}{1.3}
  \begin{tabular}{>{$}l<{$}|>{$}c<{$}|>{$}l<{$}}\toprule
    \multicolumn{1}{c|}{$4$-form}                                                                                    & \multicolumn{1}{c|}{$\partial$} & \multicolumn{1}{c}{lightcone analogue of} \\
    \midrule
  \eL_{\mathrm{LcHP}}=  \tfrac12 \epsilon_{abc} \Omega^{ab} \wedge \Theta^c                                          &                         & \text{Hilbert--Palatini}       \\
  \eL_{\mathrm{LcHolst}}= \tfrac12  \Omega_{ab} \wedge\theta^a \wedge \theta^b + \Theta_a \wedge \theta^a \wedge \xi &                         & \text{Holst}                   \\    
  \eL_{Lc\Lambda}=  \tfrac16 \epsilon_{abc} \theta^a \wedge \theta^b \wedge \theta^c \wedge \xi                      &                         & \text{Cosmological}            \\
\tfrac12 \epsilon_{abc} \Theta^a \wedge \theta^b \wedge\theta^c- \tfrac13 \eL_{\textrm{Lc}\Lambda}                   & \cm                     & \text{Pontryagin}              \\
 \tfrac14 \Omega_{ab}\wedge\Omega^{ab} - \Psi_a \wedge \Theta^a - \tfrac12 \Xi^2                                     & \cm                     & \text{Gauss--Bonnet}           \\
 \eL_{\mathrm{LcHolst}} - \tfrac12 \Theta_a \wedge \Theta^a  = -\tfrac12 d(\theta^a \wedge \Theta_a )& \cm                     & \text{Nieh--Yan}               \\
    \bottomrule
  \end{tabular}
  \caption{Summary of lightcone gauge-invariant $4$-forms and
    boundary terms ($\partial$) where
    $\Omega^{ab} = d\omega^{ab} + \omega^a{}_c \wedge \omega^{cb} +
    \psi^a \wedge \theta^b - \psi^b \wedge \theta^a$,
    $\Psi^a = d\psi^a + \omega^a{}_b \wedge \psi^b + \xi \wedge
    \psi^a$,
    $\Theta^a = d\theta^a + \omega^a{}_b \wedge \theta^b - \xi \wedge
    \theta^a$ and $\Xi = d\xi + \psi^a \wedge \theta_a$.}
  \label{tab:4-forms-LC-sum}
\end{table}
Again the Holst term has the interesting characteristic of leading
to nontrivial equations of motion that do not restrict the solution
space. The boundary terms share some similarities with their minkowskian
counterparts and show that we can exchange the Holst and cosmological
terms with $\tfrac12 \Theta_a \wedge \Theta^a$ and
$\tfrac12 \epsilon_{abc} \Theta^a \wedge \theta^b \wedge\theta^c$,
respectively.

While there is some freedom in the carrollian lagrangians we
also show the none of them describe the electric carrollian
theory~\cite{Henneaux:1979vn}. Since we have gauged all maximally
symmetric carrollian spacetimes from scratch this provides a novel
no-go result for the construction of electric theory in a first-order
formulation. See Section~\ref{sec:conclusion} for more details and a
discussion.

This paper is organised as follows. In
Section~\ref{sec:klein-geometries} we review the maximally symmetric
carrollian spacetimes. In Section~\ref{sec:gauging} we connect the
gauging procedure to Cartan geometry and describe how one can
construct gauge invariant lagrangians. We then show in
Section~\ref{sec:gauging-poincare} that the gauging procedure indeed
recovers General Relativity when applied to Minkowski spacetime. In
Section~\ref{sec:gauging-AdSC} we gauge (anti-)de~Sitter--Carroll. In
Section~\ref{sec:carroll-gauging} we gauge Carroll and also look at
how the action can be recovered as a limit from the (A)dS
counterparts. We also provide a geometric interpretation
(Section~\ref{sec:geom-interpr}). In Section~\ref{sec:gauginglc} we
gauge the lightcone. Section~\ref{sec:conclusion} provides a
conclusion with further remarks and an outlook.

\section{The Klein geometries}
\label{sec:klein-geometries}

Let $M$ be a smooth manifold on which a Lie group $G$ acts smoothly
and transitively; that is, $M$ is a homogeneous $G$-space. From the
optics of $G$, every point looks just like any other point, so we are
free to choose a point $o \in M$ and declare it to be the origin. Let
$H$ denote the subgroup of $G$ which fixes the origin. Then $M$ is
$G$-equivariantly diffeomorphic to the space $G/H$ of left cosets $gH$
for $g \in G$, where the action of $G$ on $G/H$ is induced by left
multiplication on $G$. If we let $\g$ and $\h$ denote, respectively,
the Lie algebras of $G$ and $H$, we see that that we may associate a
Klein pair $(\g,\h)$ to $M$. This turns out to capture the geometry of
$M$ as a homogeneous space of $G$ up to coverings. More precisely, as
reviewed in \cite[Appendix~B]{Figueroa-OFarrill:2018ilb}, there is a
bijective correspondence between (effective, geometrically realisable)
Klein pairs $(\g,\h)$ for a fixed $\g$ and simply-connected
homogeneous spaces of the simply-connected group $G$ with Lie algebra
$\g$. In this paper we will be applying the gauging procedure to the
Klein pairs $(\g,\h)$ corresponding to spatially isotropic carrollian
spacetimes. These were classified in \cite{Figueroa-OFarrill:2018ilb}
and studied further in \cite{Figueroa-OFarrill:2019sex} and consist of
four spacetimes: Carroll ($\zC$), de Sitter Carroll ($\zdSC$),
anti-de~Sitter--Carroll ($\zAdSC$) --- which are the carrollian limits
of the maximally symmetric lorentzian spacetimes Minkowski, de Sitter
and anti-de~Sitter, respectively --- and the lightcone ($\zLC$). They
are summarised in Table~\ref{tab:klein-pairs}, which lists the nonzero
Lie brackets of $\g = \left<L_{ab},B_a, P_a, H \right>$ where
$\h = \left<L_{ab},B_a\right>$ and where we omit those brackets
involving the rotation generators which are common to all; namely,
\begin{equation}
  \label{eq:common-brackets}
  \begin{split}
    [L_{ab},L_{cd}] &= \delta_{bc} L_{ad} - \delta_{ac} L_{bd}- \delta_{bd} L_{ac} + \delta_{ad} L_{bc}\\
    [L_{ab},B_c] &= \delta_{bc} B_a - \delta_{ac} B_b\\
    [L_{ab},P_c] &= \delta_{bc} P_a - \delta_{ac} P_b\\
    [L_{ab}, H] &= 0.
  \end{split}
\end{equation}
The table uses by now standard abbreviated notation where we drop the
vector indices, so that, e.g., $[\B,\P] = H$ stands for $[B_a,P_b]=
\delta_{ab} H$; $[H,\P] = \B$ stands for $[H,P_a] = B_a$; and $[\P,\P]
=\L$ stands for $[P_a,P_b] = L_{ab}$.

\begin{table}[h!]
  \centering
  \caption{Spatially isotropic homogeneous carrollian four-dimensional spacetimes}
  \label{tab:klein-pairs}
  \rowcolors{2}{blue!10}{white}
  \resizebox{\textwidth}{!}{
    \begin{tabular}{l|>{$}l<{$}|*{5}{>{$}l<{$}}}\toprule
      \multicolumn{1}{c|}{Name} & \multicolumn{1}{c|}{Klein pair $(\g,\h)$} & \multicolumn{5}{c}{Nonzero Lie brackets in addition to $[\L,\L] = \L$, $[\L, \B] = \B$, $[\L,\P] = \P$} \\\midrule
      Carroll                   & (\c,\iso(3))                              & [\B,\P] = H      &              &  &                &                                                   \\
      de~Sitter--Carroll        & (\iso(4), \iso(3))                        & [\B,\P] = H      & [H,\P] = -\B &  &  [\P,\P] = -\L &                                                   \\
      anti-de~Sitter--Carroll   & (\iso(3,1), \iso(3))                      & [\B,\P] = H      & [H,\P] = \B  &  &  [\P,\P] = \L  &                                                   \\
      lightcone                 & (\so(4,1), \iso(3))                       & [\B,\P] = H + \L & [H,\P] = -\P &  &                &[H,\B] = \B                                        \\\bottomrule
    \end{tabular}
  }
  \caption*{\small The Lie algebra $\g$ is different in each case and
    spanned by rotations $\L$, (carrollian) boosts $\B$, temporal
    translations $H$ and spatial translations $\P$. In all cases the
    Lie subalgebra of the Klein pair $\mathfrak{h} \cong \iso(3)$ is
    spanned $\L$ and $\B$.}
\end{table}

\section{Cartan geometry and the gauging procedure}
\label{sec:gauging}

One approach to extracting dynamics from a Klein pair, analogous to
the way that one obtains Einstein gravity departing from the Klein
pair associated to Minkowski spacetime, is via the so-called gauging
procedure \cite{Kibble:1961ba,Cho:1976fr} (see also
\cite{Andringa:2010it,Bergshoeff:2014uea,Hartong:2015zia,Hartong:2015xda}
for the application of this method to non-lorentizan geometry).
Mathematically, this is precisely the construction of a
Cartan geometry, as we will now explain.

Let us fix a Klein pair $(\g,\h)$ with $\dim \g - \dim \h = n$ and let
us fix a connected Lie group $H$ with Lie algebra $\h$.  If the Klein
pair is locally effective, so that the representation map $\h \to
\gl(\g/\h)$ is injective, then we can take $H$ to be the connected
subgroup of $\GL(\g/\h)$ generated by $\h$. More concretely, if we
pick a basis $(\overline X_1,\dots,\overline X_n)$ for $\g/\h$, where $X_i \in
\g$ and $\overline X_i = X_i \mod\h$, then we may identify $\gl(\g/\h)$
with $\gl(n,\RR)$, which is simply the Lie algebra of $n\times n$ real
matrices.  Under this identification, every $Y \in \h$ defines an $n
\times n$ real matrix which we may exponentiate to obtain an
invertible matrix $\exp Y \in \GL(n,\RR)$. We define $H$ to be the
subgroup of $\GL(n,\RR)$ consisting of finite products of such
exponentials.  It is clear from the definition that $H$ is closed
under products and inversion and also that it is connected, since if
$\gamma = \exp(Y_1) \exp(Y_2) \cdots \exp(Y_k) \in H$, then $c(t)=
\exp(tY_1) \exp(tY_2) \cdots \exp(tY_k)$ is a continuous curve in $H$
with $c(0)=\1$, the identity matrix, and $c(1)=\gamma$.

If, in addition, the Klein pair is reductive, so that
$\g = \h \oplus \m$ with $[\h,\m] \subset \m$ in the obvious notation,
then we may replace $\g/\h$ with $\m$ in the preceding paragraph and
pick the $X_i$ to be in $\m$. We shall not assume that $(\g,\h)$ is
reductive, since one of our carrollian examples (the lightcone $\zLC$)
is not reductive.

Let $M$ be a smooth manifold of dimension $n = \dim \g - \dim \h$. The
gauging procedure typically starts with a \textbf{Cartan connection}
$A$, which is a one-form $A \in \Omega^1(M,\g)$ in $M$ with values in
$\g$, i.e., $A$ is a Lie algebra $\g$-valued one-form. For $X \in \g$,
we let $\overline X \in \g/\h$ denote its image under the canonical
surjection $\g \to \g/\h$. Let
$\theta := \overline A \in \Omega^1(M,\g/\h)$ denote the projection of
$A$ to $\g/\h$. The main assumption in the gauging procedure is that
$\theta$ is a \textbf{coframe field}; that is, what is often termed an
``inverse vielbein''. The existence of such a coframe implies that the
cotangent (and hence the tangent) bundle is trivialisable and this is
only the case if $M$ is parallelisable. Since most manifolds are not
parallelisable, we need to restrict ourselves to an open subset
$U \subset M$ on which the cotangent bundle admits a trivialisation.
We can do this about every point, so let us choose an open cover
$\{U_\alpha\}$ of $M$ such that the cotangent bundle is trivialisable
over each $U_\alpha$. The gauging procedure then \emph{really} starts
with a collection $A_\alpha \in \Omega^1(U_\alpha,\g)$ of one-forms
over each $U_\alpha$, with the assumption that the projection
$\theta_\alpha \in \Omega^1(U_\alpha,\g/\h)$ of $A_\alpha$ is a local
coframe. In other words, if we choose a vector space complement $\m$
of $\h$ in $\g$ and a basis $X_i$ for $\m$, then
$\overline X_i \in \g/\h$ are a basis for $\g/\h$. Expanding
$\theta_\alpha = \theta_\alpha^i \overline X_i$, the one-forms
$\theta_\alpha^i \in \Omega^1(U_\alpha)$ are pointwise linearly
independent everywhere on $U_\alpha$. Each such pair
$(U_\alpha,A_\alpha)$ is called a \textbf{Cartan gauge} in
\cite[§5.1]{MR1453120}.

This then prompts the question of how the one-forms $A_\alpha$ and
$A_\beta$ are related in a non-empty overlap $U_\alpha \cap U_\beta$.
This is typically not discussed in the literature on the gauging
procedure, but based on the examples at our disposal, it is reasonable
to demand that on a non-empty overlap $U_\alpha \cap U_\beta$, the
one-forms $A_\alpha$ and $A_\beta$ should be related by an $H$-gauge
transformation: namely,
\begin{equation}\label{eq:overlap}
  A_\beta = \Ad(h_{\alpha\beta}^{-1}) A_\alpha + h_{\alpha\beta}^* \vartheta_H,
\end{equation}
for some smooth $h_{\alpha\beta} : U_\alpha \cap U_\beta \to H$ and
where $\vartheta_H$ is the left-invariant Maurer--Cartan form on the
group $H$.  In the case of matrix groups, which we may assume since in
our examples $H$ is a subgroup of $\GL(\g/\h)$, the above relation
says, explicitly, that for all $p \in U_\alpha\cap U_\beta$,
\begin{equation}
  A_\beta(p) = h_{\alpha\beta}(p)^{-1} A_\alpha(p) h_{\alpha\beta}(p)
  + h_{\alpha\beta}(p)^{-1} dh_{\alpha\beta}(p).
\end{equation}

In the reductive case, this equation breaks up into 
$\h$- and $\m$-components.  Writing $A_\alpha = \omega_\alpha +
\theta_\alpha$, with $\omega_\alpha \in \Omega^1(U_\alpha,\h)$ and
$\theta_\alpha \in \Omega^1(U_\alpha,\m)$,  we have that
equation~\eqref{eq:overlap} also splits into two:
\begin{subequations}
  \begin{align}
  \omega_\beta &=  \Ad(h_{\alpha\beta}^{-1}) \omega_\alpha+  h_{\alpha\beta}^* \vartheta_H\\
    \theta_\beta &= \Ad(h_{\alpha\beta}^{-1}) \theta_\alpha,
  \end{align}
\end{subequations}
from where we see that the one-forms
$\omega_\alpha \in \Omega^1(U_\alpha,\h)$ transform as the gauge
fields corresponding to an $H$-connection. In the general case, where
$\g = \h \oplus \m$ is only a vector space decomposition, the
transformation law of $\omega_\alpha$ is different and in particular
it does not transform like an $H$-connection.

Back to the general case, we cannot split equation~\eqref{eq:overlap}
into components, but we can project it to $\g/\h$.  The second term in
the RHS is $\h$-valued and hence only the first term survives the
projection, resulting in
\begin{equation}
  \theta_\beta = \overline\Ad(h_{\alpha\beta}^{-1}) \theta_\alpha,
\end{equation}
where $\overline\Ad : H \to \GL(\g/\h)$ is defined by
\begin{equation}
  \overline\Ad(h)\overline X = \overline{\Ad(h) X},
\end{equation}
for every $h \in H$ and $X \in \g$.

As shown in \cite[§5.2]{MR1453120}, and under the assumption that the
Klein geometry is effective, condition~\eqref{eq:overlap} implies
that the $\{h_{\alpha\beta}\}$ satisfy the cocycle condition on
non-empty triple overlaps; namely, for all $p \in U_\alpha \cap
U_\beta \cap U_\gamma$,
\begin{equation}
  h_{\alpha\beta}(p) h_{\beta\gamma}(p) = h_{\alpha\gamma}(p).
\end{equation}
It is then standard that the  $\{h_{\alpha\beta}\}$ are the transition
functions of a principal (right) $H$-bundle $\pi : P \to M$ and,
furthermore, that the $\{A_\alpha\}$ then assemble into a global
one-form $A \in \Omega^1(P,\g)$ on $P$ with values in $\g$ satisfying
the following three-conditions:
\begin{enumerate}
\item (\emph{non-degeneracy}) for all $p \in P$, the linear map $A_p :
  T_p P \to \g$ is an isomorphism;
\item (\emph{equivariance})  for all $h \in H$, $R_h^*A = \Ad(h^{-1})
  \circ A$, where $R_h : P \to P$ is the right action of $h \in H$ on
  $P$; and
\item (\emph{normalisation}) for all $X \in \h$, $A(\xi_X) = X$, where
  $\xi_X \in\eX(P)$ is the fundamental vector field associated to $X$,
  whose value at $p \in P$ is given by
  \begin{equation}
    \xi_X(p) = \left. \frac{d}{dt} R_{\exp(t X)}(p) \right|_{t=0}.
  \end{equation}
\end{enumerate}
The datum $(\pi: P \to M, A)$ is a \textbf{Cartan geometry on $M$
  modelled on $(\g,\h)$ with group $H$} and $A$ is said to be a
\textbf{Cartan connection}.

If the Klein geometry is not effective, so that there is a nontrivial
normal subgroup of $G$ which acts trivially on $G/H$, then the
description of a Cartan geometry via Cartan gauges is not necessarily
equivalent to the bundle description.  For ease of exposition we will
assume that the Klein geometry is effective and in any case this holds
for the examples we shall discuss in this paper.

A Cartan geometry is said to be of the \textbf{first order} if $H$ acts
faithfully on $\g/\h$.  This holds by construction in our set-up,
since $H$ has been defined as a subgroup of $\GL(\g/\h)$.  As shown in
\cite[§5.3]{MR1453120}, a first-order Cartan geometry is an
$H$-structure; that is, $P$ is a sub-bundle of the frame bundle of $M$
where frames transform under local $H$-transformations on overlaps.
If in addition, the Cartan geometry is reductive, the Cartan
connection
\begin{align}
 A = \theta + \omega
\end{align}
splits into a \textbf{soldering form} $\theta$ taking values in $\m$
and an \textbf{Ehresmann $\h$-connection} $\omega$ on $P$.

The Klein model of a Cartan geometry modelled on $(\g,\h)$ with group
$H$ is precisely the case where $P = G$, so that $M = G/H$ and $A$ is
the left-invariant Maurer--Cartan form, so that it obeys the structure
equation $dA + \tfrac12[A,A] =0$.  For a general Cartan geometry, the
\textbf{curvature} $F := dA + \tfrac12[A,A] \in \Omega^2(P,\g)$ is not
necessarily zero.  Indeed, flat Cartan geometries can be shown to be
either (open subsets in) $G/H$ or discrete quotients thereof; in other
words, they are locally isomorphic to $G/H$.  Thus we may understand
the curvature of a Cartan connection as the failure of the Cartan
geometry to be locally isomorphic to the Klein model.

Two different Klein models may give rise to the same Cartan geometry.
For example, Minkowski, de~Sitter and anti-de~Sitter spacetimes are
Klein geometries whose corresponding Cartan geometry is lorentzian
geometry. The Klein pairs corresponding to these three spacetimes are
of the form $(\g,\h)$ where $\h \cong \so(3,1)$ in all cases:
$(\iso(3,1),\so(3,1))$ for Minkowski, $(\so(4,1),\so(3,1))$ for
de~Sitter and $(\so(3,2),\so(3,1))$ for anti-de~Sitter. Although the
Lie algebras $\iso(3,1)$, $\so(4,1)$ and $\so(3,2)$ are certainly not
isomorphic, they are isomorphic as representations of the $\so(3,1)$
subalgebras in their respective Klein pairs. In Cartan geometry one
says that these Klein geometries are related by \textbf{mutation} and
the Klein geometries are said to be \textmd{model mutants} of each
other. Restricting to reductive Klein geometries, model mutation gives
rise to isomorphic Cartan geometries and there is a unique model
mutant $\g = \h \oplus \m$ which is isomorphic (as a Lie algebra) to a
semidirect product $\h \ltimes \m$ with $\m$ an abelian ideal. In the
above examples of lorentzian Klein geometries, this distinguished
mutant corresponds to Minkowski spacetime.  Cartan connections of
model mutants are of course different, since the curvature measures
the deviation of the Cartan geometry from the homogeneous model.  For
the three maximally lorentzian spacetimes above, the distinction is
essentially the cosmological constant, which is the scalar component
of the curvature in the Ricci decomposition.  Something similar will
happen with the carrollian limits of these lorentzian geometries:
$\zC$, $\zdSC$ and $\zAdSC$, which are all model mutants with Carroll
spacetime $\zC$ being the distinguished mutant.

In the reductive case, the curvature $F$ splits as $F = \Omega +
\Theta$, where $\Omega \in \Omega^2(P,\h)$ and $\Theta \in
\Omega^2(P,\m)$ are given by
\begin{equation}
  \label{eq:curv-split}
  \Theta = d_\omega \theta + \tfrac12 [\theta,\theta]_{\m}
  \qquad\text{and}\qquad
  \Omega = F_\omega + \tfrac12 [\theta,\theta]_{\h},
\end{equation}
where we used $\h$-invariance of $\m$, and where
$F_\omega := d\omega + \tfrac12 [\omega,\omega]$ and
$d_\omega \theta := d\theta + [\omega,\theta]$. By
$[\theta,\theta]_{\m}$ and $[\theta,\theta]_{\h}$ we mean the
projection of $[\theta,\theta]$ onto $\m$ and $\h$, respectively. It
should be remarked that even when the Cartan geometry is a reductive
$H$-structure, $\Omega$ need not coincide with the curvature
$F_{\omega}$ of the Ehresmann connection $\omega$. In the general
case, there is no split of $F$ and we must work with the whole
curvature.

The curvature satisfies the \textbf{Bianchi identity} $d_A F := dF +
[A,F] = 0$.  This is an immediate consequence of $d^2 =0$ and the
Jacobi identity of $\g$.  In the reductive case, the Bianchi identity
splits into two identities:
\begin{equation}
  \label{eq:bianchi-split}
  d_\omega \Omega + [\theta,\Theta]_\h = 0 \qquad\text{and}\qquad
  d_\omega \Theta + [\theta,\Omega] + [\theta,\Theta]_\m = 0,
\end{equation}
where we in general define the action of $d_{\omega}$ as
$d_{\omega}\alpha := d \alpha + [\omega,\alpha]$ and where the Lie
bracket hides a wedge.

Locally on each open subset $U_\alpha$ we have $F_\alpha = dA_\alpha +
\tfrac12 [A_\alpha,A_\alpha] \in \Omega^2(U_\alpha,\g)$, which again
satisfies the Bianchi identity $dF_\alpha + [A_\alpha,F_\alpha] = 0$.
It follows from equation~\eqref{eq:overlap} that on non-empty
$U_\alpha \cap U_\beta$,
\begin{equation}
  F_\beta = \Ad(h_{\alpha\beta}^{-1}) F_\alpha = h_{\alpha\beta}^{-1}
  F_{\alpha} h_{\alpha\beta},
\end{equation}
where the second expression holds for matrix groups.

Since both $\{\theta_\alpha\}$ and $\{F_\alpha\}$ transform linearly
under $H$, any $n$-form constructed out of wedge products of
$\theta_\alpha$ and $F_\alpha$ contracted with $H$-invariant tensors
will agree on overlaps and hence will glue to a global $n$-form on
$M$. This is how we will construct gauge-invariant
\textbf{lagrangians} $\eL \in \Omega^n(M)$ associated to a Cartan
geometry. In odd dimensions, one could also add Chern--Simons terms,
which are invariant under infinitesimal gauge transformations and
built out of the $\omega_{\alpha}$. In this paper we are interested in
four-dimensional theories and hence will not need to consider
Chern--Simons terms.

Varying the lagrangian $\eL$ with respect to $A$ and using
the Bianchi identities to eliminate derivatives on the curvatures,
results in algebraic equations for the curvatures and $\theta$ in the
free differential graded commutative algebra (dgca) generated by
$\theta$.  We record here for later use the variation of the curvature
\begin{equation}
  \delta F = d_A \delta A = d\delta A + [A,\delta A].
\end{equation}
In the reductive case, this splits as
\begin{equation}\label{eq:variations}
  \delta\Omega = d_\omega \delta\omega + [\theta,\delta\theta]_\h
  \qquad\text{and}\qquad
  \delta\Theta = d_\omega \delta\theta + [\theta,\delta\omega] + [\theta,\delta\theta]_\m.
\end{equation}

We now specialise to the cases of interest, which are all in dimension
$n=4$.  We seek $H$-gauge-invariant $4$-forms made out of $F$ and
$\theta$.  Schematically we have three kinds of $4$-forms, illustrated
in Table~\ref{tab:4-forms-nonred} along with the $H$-representations
they belong to.  To arrive at $H$-gauge invariant $4$-forms, we need
to contract them with $H$-invariant tensors in the duals of the
representations in the table.  This will require determining the
$H$-invariant tensors in those representations.  Since we assume that
$H$ is connected, the Lie correspondence guarantees that it is
sufficient to determine the $\h$-invariant tensors, which is a much
easier (linear-algebraic) problem.

\begin{table}[h!]
  \centering
  \begin{tabular}{>{$}l<{$}|>{$}l<{$}}\toprule
    \multicolumn{1}{c|}{$4$-form} & \multicolumn{1}{c}{$\h$-representation}\\
    \midrule
    F \wedge F & \odot^2\g\\
    F \wedge \theta \wedge \theta & \g \otimes \wedge^2 (\g/\h)\\
    \theta \wedge\theta \wedge \theta \wedge\theta&  \wedge^4 (\g/\h)\\
    \bottomrule
  \end{tabular}
  \caption{Four-forms (schematically) and their representations}
  \label{tab:4-forms-nonred}
\end{table}

In the reductive case, we may refine the above discussion. We now
construct the lagrangian $4$-form out of the one-form $\theta$ and the
two-forms $\Omega$ and $\Theta$ into which the curvature splits.
Schematically, the possible $4$-forms we can make up out of these
ingredients are given in Table~\ref{tab:4-forms} together with the
$\h$-representation they belong to. As before, the $H$-gauge invariant
four-forms are constructed by contracting the terms in the table with
$H$-invariant tensors in the duals of the representations in the
table.

\begin{table}[h!]
  \centering
  \begin{tabular}{>{$}l<{$}|>{$}l<{$}}\toprule
    \multicolumn{1}{c|}{$4$-form} & \multicolumn{1}{c}{$\h$-representation}\\
    \midrule
    \Omega \wedge \Omega & \odot^2\h\\
    \Omega \wedge \Theta & \h \otimes \m\\
    \Theta \wedge \Theta & \odot^2\m\\
    \theta \wedge\theta \wedge \Omega&  \wedge^2 \m \otimes \h\\
    \theta \wedge\theta \wedge \Theta&  \wedge^2 \m \otimes \m\\
    \theta \wedge\theta \wedge \theta \wedge\theta&  \wedge^4 \m\\    
    \bottomrule
  \end{tabular}
  \caption{Four-forms (schematically) and their representations
    (reductive)}
  \label{tab:4-forms}
\end{table}

The construction of the gauge-invariant lagrangian only sees how the
various objects (soldering form, curvature,...) transform under $H$.
Therefore provided we construct the most general gauge-invariant
lagrangian, the resulting lagrangian will be formally invariant under
model mutation, even though the explicit expression of the curvatures,
the Bianchi identities and the variations will not be, as they depend
on the explicit Lie brackets of $\g$.  Therefore it is perhaps not
clear \emph{a priori} that the resulting lagrangians are physically
equivalent.

The three Klein models corresponding to Carroll ($\zC$),
de~Sitter--Carroll ($\zdSC$) and anti-de~Sitter--Carroll ($\zAdSC$)
are related by mutation, with $\zC$ being the distinguished mutant
with the abelian ideal.  As a check of our contention that the gauging
procedure is really just constructing a Cartan geometry, we will show
that the resulting carrollian gravity theories modelled on $\zC$,
$\zdSC$ and $\zAdSC$, are equivalent.

Let us now turn to the gauging method applied to the Klein pairs of
(anti-)~de~Sitter--Carroll, flat Carroll and the lightcone,  but to
ease ourselves into the calculations, we will first review the
classical case of Minkowski spacetime.

\section{``Gauging the Poincaré algebra''}
\label{sec:gauging-poincare}

In this section we review the gauging procedure starting from the
Klein pair $(\g,\h)$ of (four-dimensional) Minkowski spacetime, where
$\g$ is the Poincaré algebra and $\h$ is a Lorentz subalgebra.  This
is often described as ``gauging the Poincaré algebra'', but this is
clearly imprecise.  The Klein pair $(\g,\h)$ of $\zAdSC$ also has $\g$
being isomorphic to the Poincaré algebra, but crucially it is $\h$
which differs.  A more precise description would be that we are
gauging a Klein pair $(\g,\h)$; although it should be clear from the
description above that, if anything is being gauged, it is actually
the subalgebra $\h$.

We shall choose a basis $L_{mn} = - L_{nm}$ and $P_m$, where
$m=0,\dots,3$, for $\g$ adapted to the reductive split
$\g = \h \oplus \m$, with $\h = \left<L_{mn}\right>$ and
$\m = \left<P_m\right>$. In this basis, $\g$ comes with the following
brackets
\begin{equation}\label{eq:poincare}
  \begin{split}
    [L_{mn},L_{pq}] &= \eta_{np}L_{mq} - \eta_{mp}L_{nq} - \eta_{nq}L_{mp} + \eta_{mq}L_{np} \\
    [L_{mn},P_q] &= \eta_{nq}P_m - \eta_{mq}P_n\\
    [P_m,P_n] &= 0,
  \end{split}
\end{equation}
where we take $\eta_{mn}$ to be mostly plus.  The split $\g = \h
\oplus \m$, is not just reductive, but also symmetric $([\m,\m]
\subset \h)$.  We also observe that the Poincaré Lie algebra is
$\ZZ$-graded, with $L_{mn}$ in degree $0$ and $P_m$ in degree $-1$,
say, as is typical in geometric applications.

The restriction to $\h$ of the adjoint representation defines an
injective map $\h \to \gl(\m)$ and we let $H < \GL(\m)$ denote the
connected Lie group the image of $\h$ generates.  By a classic result
of Weyl's \cite[Theorem~ 2.11.A]{MR1488158}, all $H$-invariant tensors
are made out of $\eta_{mn}$ and $\epsilon_{mnpq}$.

Let $M$ be a four-dimensional smooth manifold.  We shall work locally
in an open subset $U \subset M$ with the tacit assumption that we are
repeating these calculations on each member of an open cover for $M$
with gluing conditions as explained in
Section~\ref{sec:gauging}, particularly
equation~\eqref{eq:overlap}.  We shall let
$A = \tfrac12 \omega^{mn} L_{mn} + \theta^m P_m \in \Omega^1(U,\g)$ and
$F = \tfrac12 \Omega^{mn} L_{mn} + \Theta^m P_m \in \Omega^2(U,\g)$.  Explicitly, we have
\begin{equation}
\label{eq:lorentzian-curvatures}
  \begin{split}
    \Omega^{mn} &= d\omega^{mn} +\omega^m{}_p \wedge \omega^{pn}\\
    \Theta^m &= d^\nabla \theta^m := d\theta^m + \omega^m{}_n \wedge \theta^n,
  \end{split}
\end{equation}
where indices are lowered with $\eta_{mn}$. The Bianchi identities take the form
\begin{equation}\label{eq:Bianchi-identity-Poincare}
  d^\nabla \Omega^{mn} = 0 \qquad\text{and}\qquad d^\nabla \Theta^m =
  \Omega^m{}_n \wedge \theta^n.
\end{equation}
where the action of $d^{\nabla}$ is defined by
\begin{equation}
    d^\nabla \alpha^{mn} := d\alpha^{mn} + \omega^m{_p}\wedge \alpha^{pn} + \omega^n{_p}\wedge\alpha^{mp}.
\end{equation}

The space of $\h$-invariant four-forms constructed out of
$\Omega^{mn},\Theta^m,\theta^m$ is six-dimensional and spanned by the
four-forms in Table~\ref{tab:4-forms-mink}, which also records their $\h$-representation and their
Lie algebraic degree relative to the grading by which $L_{mn}$ has
degree $0$ and $P_m$ has degree $-1$.

\begin{table}[h!]
  \centering
    \rowcolors{2}{blue!10}{white}
  \renewcommand{\arraystretch}{1.3}
  \begin{tabular}{>{$}l<{$}|>{$}l<{$}|>{$}l<{$}}\toprule
    \multicolumn{1}{c|}{$4$-form} & \multicolumn{1}{c|}{degree} & \multicolumn{1}{c}{$\h$-representation}\\
     \midrule
    \tfrac18 \epsilon_{mnpq} \Omega^{mn} \wedge \Omega^{pq} & 0 & \odot^2\h\\
    \tfrac14 \Omega_{mn} \wedge \Omega^{mn} & 0& \odot^2\h\\
    \tfrac12 \Theta_m \wedge \Theta^m & 2& \odot^2\m\\
    \tfrac12 \Omega_{mn} \wedge \theta^m \wedge \theta^n & 2& \h \otimes \wedge^2\m\\
    \tfrac14 \epsilon_{mnpq} \Omega^{mn} \wedge \theta^p \wedge \theta^q & 2& \h \otimes \wedge^2\m\\
    \tfrac1{24} \epsilon_{mnpq} \theta^m \wedge \theta^n \wedge \theta^p \wedge \theta^q& 4& \wedge^4\m\\
    \bottomrule
  \end{tabular}
  \caption{Gauge-invariant $4$-forms}
  \label{tab:4-forms-mink}
\end{table}

Under the variations of $\omega^{mn}$ and $\theta^m$, the curvatures
vary as follows:
\begin{equation}
  \delta\Omega^{mn} = d^\nabla \delta\omega^{mn}
  \qquad\text{and}\qquad
  \delta\Theta^m = d^\nabla \delta\theta^m + \delta \omega^m{}_n
  \wedge \theta^n.
\end{equation}
Under such variations, the top two terms in the table vary into exact
forms upon using the Bianchi identity
\eqref{eq:Bianchi-identity-Poincare}:
\begin{equation}
  \begin{split}
    \delta( \tfrac18 \epsilon_{mnpq}\Omega^{mn} \wedge \Omega^{pq}) &=  d \left( \tfrac14 \epsilon_{mnpq} \Omega^{mn} \wedge  \delta\omega^{pq} \right)\\
    \delta( \tfrac14 \Omega_{mn} \wedge \Omega^{mn}) &= d(\tfrac12 \Omega_{mn} \wedge \delta\omega^{mn}),
  \end{split}
\end{equation}
and hence neither of these two terms contributes to the
Euler--Lagrange equations.  Similarly, from the variations of the
middle two terms:
\begin{equation}
  \begin{split}
    \delta(\tfrac12 \Theta^m \wedge \Theta_m) &= d(\Theta_m \wedge \delta\theta^m)  + \Omega_{mn} \wedge \theta^m \wedge \delta\theta^n  - \Theta^m \wedge \theta^n \wedge \delta\omega_{mn}\\
    \delta(\tfrac12 \Omega_{mn} \wedge \theta^m \wedge \theta^n) &=
    d(\tfrac12 \delta\omega_{mn} \wedge \theta^m \wedge \theta^n) -
    \Theta^m \wedge \theta^n \wedge \delta\omega_{mn} + \Omega_{mn}
    \wedge \theta^m \wedge \delta\theta^n,
  \end{split}
\end{equation}
we see that their difference varies into an exact form and hence we
need only keep one of them in the lagrangian. Explicitly, if two
$4$-forms $X,Y\in\Omega^4(M)$ are such that their difference varies
into an exact form, $\delta(X-Y) = dZ$ for some $Z\in\Omega^3(M)$,
then the equations of motion obtained from the lagrangian
$\eL = aX + bY$ are the same as those obtained from $\eL' = c X$,
where $a,b$ and $c$ are real parameters. To see this, simply write
$\eL = (a+b)X + b(X-Y)$, so that after identifying $c = a+b$, we get
$\delta (\eL - \eL') = dZ$. We have summarised the terms with
nontrivial and trivial variation in Table~\ref{tab:4-forms-mink-sum}.

\begin{table}[h!]
  \centering
    \rowcolors{2}{blue!10}{white}
  \renewcommand{\arraystretch}{1.3}
  \begin{tabular}{>{$}l<{$}|>{$}c<{$}|>{$}l<{$}}\toprule
    \multicolumn{1}{c|}{$4$-form}                                                                                              & \multicolumn{1}{c|}{$\partial$} & \multicolumn{1}{c}{}     \\
     \midrule
    \tfrac14 \epsilon_{mnpq} \Omega^{mn} \wedge \theta^p \wedge \theta^q                                                       &                         & \text{Hilbert--Palatini} \\
    \tfrac1{24} \epsilon_{mnpq} \theta^m \wedge \theta^n \wedge \theta^p \wedge \theta^q                                       &                         & \text{Cosmological}      \\
    \tfrac12 \Omega_{mn} \wedge \theta^m \wedge \theta^n \text{ or }\tfrac12  \Theta_m \wedge \Theta^m                         &                         & \text{Holst}             \\
    \tfrac18 \epsilon_{mnpq} \Omega^{mn} \wedge \Omega^{pq}                                                                    & \cm                     & \text{Pontryagin}        \\
    \tfrac14 \Omega_{mn} \wedge \Omega^{mn}                                                                                    & \cm                     & \text{Gauss--Bonnet}     \\
    \tfrac12 ( \Theta_m \wedge \Theta^m - \Omega_{mn} \wedge \theta^m \wedge \theta^n) =  \tfrac12 d(\theta^m \wedge \Theta_m) & \cm                     & \text{Nieh--Yan}         \\
    \bottomrule
  \end{tabular}
  \caption{Summary of gauge-invariant $4$-forms and boundary terms
    ($\partial$). The Nieh--Yan term shows that we could replace the Holst
    term by $\Theta_m \wedge \Theta^m$.}
  \label{tab:4-forms-mink-sum}
\end{table}

In summary, introducing real parameters $\mu$ and $\Lambda$, the
lagrangian can be taken to be
\begin{equation}
\label{eq:relativistic-lagrangian}
  \eL = \tfrac14 \epsilon_{mnpq} \theta^m \wedge \theta^n \wedge
  \Omega^{pq} + \tfrac\beta2 \theta^m \wedge \theta^n \wedge \Omega_{mn}
  - \tfrac{\Lambda}{4!} \epsilon_{mnpq} \theta^m \wedge 
  \theta^n \wedge \theta^p \wedge \theta^q.
\end{equation}
This expression is $H$-gauge invariant by construction and hence gives
rise to a global $4$-form $\eL \in \Omega^4(M)$, whose first term is
the standard Hilbert--Palatini lagrangian, the second term is the
modification due to Holst \cite{Holst:1995pc} and the third term is a
cosmological constant term. It is interesting to note that one could
replace the Holst term with $\Theta^m \wedge \Theta_m$. While this
does not alter the equations of motion, it is a change of the boundary
structure which might be relevant when there are boundaries at finite
distance or corner symmetries.
  
The variation of this lagrangian gives rise to two equations:
\begin{subequations}
  \begin{align}
    \tfrac12 \epsilon_{mnpq} \Theta^p \wedge \theta^q  &= - \tfrac\beta2 \left(\Theta_m \wedge \theta_n - \Theta_n \wedge \theta_m\right)\label{eq:var-omega}\\
    \tfrac12 \epsilon_{mnpq} \theta^m \wedge \Omega^{np} &= \beta  \Omega_{qp} \wedge \theta^p + \tfrac{\Lambda}{6} \epsilon_{mnpq} \theta^m \wedge \theta^n \wedge \theta^p.\label{eq:var-theta}
  \end{align}
\end{subequations}

\begin{lemma}\label{lem:torsion-free}
  Equation~\eqref{eq:var-omega} implies that $\Theta^m = 0$ for any
  (real) value of $\beta$.
\end{lemma}

\begin{proof}
  Consider $\Psi := \tfrac12 (\Theta^m \wedge \theta^n - \Theta^n
  \wedge \theta^m) P_m \wedge P_n \in \Omega^3(U,\wedge^2\m)$.  Let
  $\star: \wedge^2\m \to \wedge^2\m$ be the internal Hodge star,
  defined by
  \begin{equation}
    (\star \Phi)^{mn} = \tfrac12 \epsilon^{mnpq} \Phi_{pq},
  \end{equation}
  for any $\wedge^2\m$-valued object $\Phi$.  It is evident that
  $\star$ operator involves the inner product on $\m$ and since this
  is lorentzian, we have that $\star^2 = - \id_{\wedge^2\m}$.  Now
  equation~\eqref{eq:var-omega} says that $\star \Psi = \beta \Psi$ and
  applying $\star$ again we see that $(\beta^2+1)\Psi = 0$.  Hence if
  $\beta$ is real, the only solution is $\Psi = 0$,
  which is equivalent to
  \begin{equation}\label{eq:aux-lem-var-omega}
    \epsilon_{mnpq} \Theta^p \wedge \theta^q = 0.
  \end{equation}
  We may expand $\Theta^p= \tfrac12 \Theta^p{}_{rs} \theta^r \wedge
  \theta^s$ and hence this becomes
  \begin{equation}
    \epsilon_{mnpq} \Theta^p{}_{rs} \theta^q \wedge \theta^r \wedge
    \theta^s = 0,
  \end{equation}
  or equivalently
  \begin{equation*}
    \epsilon_{mnp[q}\Theta^p_{rs]} = 0.
  \end{equation*}
  Contracting with $\epsilon^{trsq}$ and using the standard identity
  \begin{equation}\label{eq:faulknerism}
    \epsilon^{rstq}\epsilon_{mnpq} = 6 \delta^{[r}_m \delta^s_n \delta^{t]}_p,
  \end{equation}
  we arrive at
  \begin{equation*}
    \delta^s_m \Theta^p_{nq} +     \delta^s_n \Theta^p_{qm} +
    \delta^s_q \Theta^p_{mn} = 0.
  \end{equation*}
  Tracing with $\delta^m_s$, we see that $\Theta^p_{nq}=0$, proving
  the lemma.
\end{proof}

Since $\Theta = 0$, the second of the Bianchi identities in
\eqref{eq:Bianchi-identity-Poincare} becomes $\Omega_{mn} \wedge
\theta^n = 0$, which kills the $\beta$-dependent term in
equation~\eqref{eq:var-theta}.  Let us write
\begin{equation}
  \Omega_{mn} = - \tfrac12 R_{mnpq} \theta^p \wedge \theta^q
\end{equation}
and insert this into equation~\eqref{eq:var-theta} to obtain
\begin{equation}
  \label{eq:einstein}
  \tfrac12 R^{mn}{}_{mn} \eta_{pq} + R_{pmnq} \eta^{mn} + \Lambda \eta_{pq} = 0.
\end{equation}
Let us define the Ricci tensor $R_{pq} := \eta^{mn} R_{pmnq}$ and the
Ricci scalar $R = \eta^{pq} R_{pq}$, in terms of which the above
equation becomes
\begin{equation}
  \label{eq:einstein-2}
  \tfrac12 R \eta_{pq} - R_{pq} - \Lambda \eta_{pq} = 0.
\end{equation}
Tracing with $\eta^{pq}$, we see that $R = 4 \Lambda$ and
re-inserting this into equation~\eqref{eq:einstein-2} we arrive at
\begin{equation}
  \label{eq:vacuum-einstein-cosmo}
  R_{pq} = \Lambda \eta_{pq}.
\end{equation}

By construction, the Cartan geometry under discussion, being of the
first order, is an $H$-structure.  Since $H$ leaves invariant a
lorentzian inner product on $\m$, we get a lorentzian metric on $M$
which has local expression $g = \eta_{mn}\theta^m\theta^n$ which
allows us to identify the bundle $P\to M$ as the bundle of oriented
(and time-oriented) orthonormal frames.  Being reductive, the
connection $\omega$ takes values in $\h$ and hence it is
metric-compatible and Lemma~\ref{lem:torsion-free} says that it is
torsion-free, hence it is essentially the Levi-Civita connection and
$R_{mnpq}$ are the coefficients of the Riemann tensor relative to the
local orthonormal frame $\theta^m$, with $R_{mn}$ and $R$ the standard
Ricci tensor and Ricci scalar, so that
equation~\eqref{eq:vacuum-einstein-cosmo} is indeed the vacuum Einstein
equation with a cosmological constant.

Let us point out that by tuning the cosmological constant, we can
arrange for Minkowski, de~Sitter or anti-de~Sitter spacetimes to solve
the Euler--Lagrange equations. If we didn't have the cosmological
constant term, only Minkowski spacetime would solve them. In order to
find de~Sitter or anti-de~Sitter spacetimes as solutions, we would
have to build a Cartan geometry modelled on these spacetimes. In
\cite{Wise:2006sm} it is shown that doing so one recovers
MacDowell--Mansouri gravity. It is not uncommon to impose on the
gauge-invariant lagrangians the condition that the Klein model should
be a solution. In the case of the models being the maximally symmetric
lorentzian manifolds, the difference between their curvature tensors
is just the scalar component in the Ricci decomposition and this
precisely corresponds to the cosmological constant term. Therefore by
tuning the cosmological constant parameter in the lagrangian we can
obtain Minkowski, de~Sitter or anti-de~Sitter spacetimes (with any
radius of curvature) as solutions.  It is not clear to us whether this
holds for more general model mutations; that is, whether there is a
modification of the lagrangian which can exhibit any model mutant as a
solution of its Euler--Lagrange equation.

\section{Gauging the Klein pairs of $\zdSC$ and $\zAdSC$}
\label{sec:gauging-AdSC}

Let us apply the gauging procedure to $\zdSC$ and $\zAdSC$.  We will
treat both cases simultaneously since, as we shall see, \emph{mutatis
  mutandis}, the Euler--Lagrange equations take the same form.  The
Klein pairs $(\g_{\zdSC},\h)$ and $(\g_{\zAdSC},\h)$ are described in
Table~\ref{tab:klein-pairs}: $\g_{\zdSC} = \iso(4)$, the euclidean Lie
algebra, whereas $\g_{\zAdSC} = \iso(3,1)$, the Poincaré Lie algebra.
We can introduce a sign $\sigma = \pm 1$ and write the nonzero Lie
brackets of both Lie algebras as follows:
\begin{equation}
  \label{eq:dsc-adsc-combined-brackets}
  \begin{split}
    [L_{ab},L_{cd}] &= \delta_{bc} L_{ad} - \delta_{ac} L_{bd}- \delta_{bd} L_{ac} + \delta_{ad} L_{bc}\\
    [L_{ab},B_c] &= \delta_{bc} B_a - \delta_{ac} B_b\\
    [L_{ab},P_c] &= \delta_{bc} P_a - \delta_{ac} P_b\\
    [B_a,P_b] &= \delta_{ab} H\\
    [H, P_a] &= \sigma B_a\\
    [P_a, P_b] &= \sigma L_{ab},
  \end{split}
\end{equation}
where $\sigma = +1$ for $\zAdSC$ and $\sigma = -1$ for $\zdSC$.
Taking $\sigma = 0$ we arrive at the Carroll algebra for $\zC$, but in
this section we will assume that $\sigma \neq 0$.

In both of these cases, the connection $A$ will turn out to be a
Cartan connection in a first-order, reductive Cartan geometry.
Therefore we could arrive at equivalent equations departing from any
other Klein model for the carrollian Cartan geometry.  There
is a unique such model in which $\g = \h \ltimes \m$, with $\m$
abelian.  The resulting Klein pair $(\h \ltimes \m, \h)$ describes the
Carroll spacetime, since $\h \ltimes \m \cong \c$, the Carroll
algebra.  The resulting calculations are technically simpler and,
formally, they can be obtained from ours by setting $\sigma = 0$.  We
will discuss this case separately in
Section~\ref{sec:carroll-gauging}.

\subsection{The gauge fields}
\label{sec:gauge-fields}

We consider an open subset $U \subset M$ of a four-dimensional smooth
manifold $M$ and introduce a one-form $A \in \Omega^1(U,\g)$
with values in $\g$ which, relative to the above basis, can be
expanded as
\begin{equation}
  \label{eq:cartan-connection}
  A = \tfrac12 \omega^{ab}L_{ab} + \psi^a B_a + \theta^a P_a + \xi H.
\end{equation}
Its curvature $F\in\Omega^2(U,\g)$ is given by
\begin{equation}
  \label{eq:cartan-curvature}
  F = dA + \tfrac12 [A,A] = \tfrac12 \Omega^{ab} L_{ab} + \Psi^a B_a + \Theta^a P_a + \Xi H,
\end{equation}
where
\begin{equation}
  \label{eq:curvatures}
  \begin{split}
    \Omega^{ab} &= d\omega^{ab} + \omega^a{}_c \wedge \omega^{cb} + \sigma \theta^a \wedge \theta^b = F_\nabla^{ab} + \sigma \theta^a \wedge \theta^b\\
    \Psi^a &= d\psi^a + \omega^a{}_b \wedge \psi^b + \sigma \xi \wedge \theta^a = d^\nabla \psi^a + \sigma \xi \wedge \theta^a\\
    \Theta^a &= d\theta^a + \omega^a{}_b \wedge \theta^b = d^\nabla \theta^a\\
    \Xi &= d\xi + \psi^a \wedge \theta_a,
  \end{split}
\end{equation}
where $d^\nabla$ is the $\r$-covariant exterior derivative and
$F_\nabla$ its curvature two-form.  Indices have been lowered
with $\delta_{ab}$.  The Bianchi identity $dF + [A,F]=0$ also splits as
follows:
\begin{equation}
  \label{eq:bianchi-adsc-split}
  \begin{split}
    d^\nabla\Omega^{ab} &=  \sigma \left( \Theta^a \wedge \theta^b - \theta^a \wedge \Theta^b \right)\\
    d^\nabla \Psi^a &= \Omega^a{}_b \wedge \psi^b + \sigma \theta^a \wedge \Xi - \sigma \xi \wedge \Theta^a\\
    d^\nabla\Theta^a &=  \Omega^a{}_b \wedge \theta^b\\
    d \Xi &= \theta^a \wedge \Psi_a - \psi^a \wedge \Theta_a,
  \end{split}
\end{equation}
where $d^\nabla \Omega^{ab} = d\Omega^{ab} + \omega^a{}_c \wedge \Omega^{cb}- \omega^b{}_c \wedge \Omega^{ca}$.

We now construct the lagrangian $4$-form out of the two-forms $\Omega
= (\Omega^{ab},\Psi^a)$ and $\Theta = (\Theta^a,\Xi)$ and the
one-forms $\theta = (\theta^a,\xi)$.  The possible $4$-forms we can
make up out of these ingredients are given in
Table~\ref{tab:4-forms} and we must contract them with $H$-invariant
tensors in the relevant representations.  We now turn to their
determination.

\subsection{Invariant tensors}
\label{sec:invariant-tensors}

The duals of the $\h$-representations appearing in
Table~\ref{tab:4-forms} are subspaces of the tensor algebra of $\g^*$.
The basis $(L_{ab}, B_a, P_a, H)$ for $\g$ induces a canonically dual
basis $(\lambda^{ab}, \beta^a,\pi^a,\eta)$ for $\g^*$ satisfying 
\begin{equation}
\label{eq:dual-basis}
\braket{L_{ab},\lambda^{cd}} = \delta^c_a\delta^d_b - \delta^d_a\delta^c_b, \qquad \braket{B_a,\beta^b} = \delta^b_a,\qquad \braket{P_a,\pi^b} = \delta_a^b,\qquad \braket{H,\eta} =1. 
\end{equation}
It is relatively easy to determine the $\h$-invariant tensors in $\odot^2\h^*$,
$\odot^2\m^*$, $\h^*\otimes \m^*$, $\wedge^2\m^* \otimes \h^*$,
$\wedge^2\m^* \otimes \m^*$ and $\wedge^4\m^*$.  We first write down
the $\r$-invariant tensors, which are made out of $\delta_{ab}$ and
$\epsilon_{abc}$, and then we check which linear combinations of
$\r$-invariant tensors are also invariant under the coadjoint action
of the $B_a$.  This action is given by $B_a \cdot \alpha = - \alpha
\circ \ad_{B_a}$ for any $\alpha \in \g^*$ and we find
\begin{equation}
  \label{eq:coadjoint-action-B}
  B_a \cdot \lambda^{bc} = 0, \qquad B_a \cdot \beta^b =
  \lambda^{b}{_a}, \qquad B_a \cdot \pi^b = 0
  \qquad\text{and}\qquad B_a \cdot \eta = - \pi_a.
\end{equation}
Let us go through each of the representations in
Table~\ref{tab:4-forms} in turn. Firstly,~~$\wedge^4\m^*$ is
one-dimensional and spanned by
$\tfrac16 \epsilon_{abc}\pi^a \wedge \pi^b \wedge \pi^c \wedge \eta$,
which is manifestly $\h$-invariant. This tensor contributes a term
$\tfrac16 \epsilon_{abc} \theta^a \wedge \theta^b \wedge \theta^c
\wedge \xi$ to the lagrangian.

There are two $\r$-invariant tensors in $\wedge^2 \m^* \otimes \m^*$:
\begin{equation}
  \tfrac12 \epsilon_{abc} \pi^a \wedge \pi^b \otimes \pi^c
  \qquad\text{and}\qquad \delta_{ab} \eta \wedge \pi^a \otimes \pi^b.
\end{equation}
The first term is invariant under the coadjoint action of $B_a$,
whereas the second term is not.  The corresponding term in the
lagrangian is $\tfrac12 \epsilon_{abc} \theta^a \wedge \theta^b
\wedge \Theta^c$.

There are four $\r$-invariant tensors in $\wedge^2\m^* \otimes \h^*$:
\begin{equation}
  \tfrac12 \pi^a \wedge \pi^b \otimes \lambda_{ab},
  \qquad \tfrac12 \epsilon_{abc} \pi^a \wedge \pi^b \otimes \beta^c,
  \qquad \eta \wedge \pi^a \otimes \beta_a
  \qquad\text{and}\qquad
  \tfrac12 \epsilon_{abc} \eta \wedge \pi^a \otimes \lambda^{bc},
\end{equation}
where we have lowered some indices using $\delta_{ab}$.  The first
term is clearly invariant under $B_a$ and there is a linear
combination of two of the remaining terms which is also invariant, namely
$\tfrac12 \epsilon_{abc} (\pi^a \wedge \pi^b \otimes \beta^c + \eta
\wedge \pi^a \otimes \lambda^{bc})$, while the penultimate term in the list above is not invariant.  These invariant tensors
contribute the following terms to the lagrangian $\tfrac12 \theta^a
\wedge \theta^b \wedge \Omega_{ab}$ and $\tfrac12 \epsilon_{abc}
\left(\theta^a \wedge \theta^b \wedge \Psi^c + \xi \wedge \theta^a
  \wedge \Omega^{bc}\right)$, respectively.

There are two $\r$-invariant tensors in $\odot^2\m^*$: $
\pi_a \pi^b$ and $\eta^2$.  The former is also invariant under $B_a$,
whereas the latter is not.  This representation contributes one term
to the lagrangian, namely $\tfrac12 \Theta^a \wedge \Theta_a$.

There are two $\r$-invariant tensors in $\m^* \otimes \h^*$, namely
$\tfrac12 \epsilon_{abc} \pi^a \otimes \lambda^{bc}$ and
$\pi^a \otimes \beta_a$. The former is invariant under $B_a$, whereas
the latter is not. In summary, we have one possible new term in the
lagrangian, namely
$\tfrac12 \epsilon_{abc} \Theta^a \wedge \Omega^{bc}$.

Finally, there are three $\r$-invariant tensors in $\odot^2\h^*$:
\begin{equation}
  \tfrac12 \lambda^{ab}\lambda_{ab}, \qquad \beta^a \beta_a
  \qquad\text{and}\qquad
  \tfrac12 \epsilon_{abd}\beta^a \lambda^{bc}.
\end{equation}
Of these, only the first and last are also invariant under $B_a$,
contributing two terms to the lagrangian: $\tfrac14 \Omega^{ab} \wedge
\Omega_{ab}$ and $\tfrac12 \epsilon_{abc} \Psi^a \wedge \Omega^{bc}$.

Table~\ref{tab:invariant-tensors} summarises the above discussion,
where we also record the degree relative to the grading described
above of $\g$.  Recall that $L_{ab}$ and $P_a$ have zero degree,
whereas $B_a$ and $H$ have degree $\alpha$.  Dually, $\lambda^{ab},
\pi^a$ have degree $0$ whereas $\beta^a$  and $\eta$ have degree
$-\alpha$.

\begin{table}[h!]
  \centering
    \rowcolors{2}{blue!10}{white}
  \renewcommand{\arraystretch}{1.3}
  \begin{tabular}{>{$}l<{$}|>{$}l<{$}|>{$}r<{$}}\toprule
    \multicolumn{1}{c|}{invariant tensor}                                                                 & \multicolumn{1}{c|}{$4$-form}                                                                                        & \multicolumn{1}{c}{degree} \\
    \midrule
    \tfrac16 \epsilon_{abc}\pi^a \wedge \pi^b \wedge \pi^c \wedge \eta                                    & \tfrac16 \epsilon_{abc} \theta^a \wedge \theta^b \wedge \theta^c \wedge \xi                                          & -\alpha                    \\ 
    \tfrac12 \epsilon_{abc} \pi^a \wedge \pi^b \otimes \pi^c                                              & \tfrac12 \epsilon_{abc} \theta^a \wedge \theta^b \wedge \Theta^c                                                     & 0                          \\ 
    \tfrac12 \pi^a \wedge \pi^b \otimes \lambda_{ab}                                                      & \tfrac12 \theta^a \wedge \theta^b \wedge \Omega_{ab}                                                                 & 0                          \\ 
    \tfrac12 \epsilon_{abc} (\pi^a \wedge \pi^b \otimes \beta^c + \eta \wedge \pi^a \otimes \lambda^{bc}) & \tfrac12 \epsilon_{abc} \left(\theta^a \wedge \theta^b \wedge \Psi^c + \xi \wedge \theta^a \wedge \Omega^{bc}\right) & - \alpha                   \\ 
    \pi^a \pi_a                                                                                           & \tfrac12 \Theta^a \wedge \Theta_a                                                                                    & 0                          \\
    \tfrac12 \epsilon_{abc} \pi^a \otimes \lambda^{bc}                                                    & \tfrac12 \epsilon_{abc} \Theta^a \wedge \Omega^{bc}                                                                  & 0                          \\
    \tfrac12 \lambda^{ab}\lambda_{ab}                                                                     & \tfrac14 \Omega^{ab} \wedge \Omega_{ab}                                                                              & 0                          \\
    \tfrac12 \epsilon_{abc}\beta^a \lambda^{bc}                                                           & \tfrac12 \epsilon_{abc} \Psi^a \wedge \Omega^{bc}                                                                    & -\alpha                    \\
  \bottomrule
  \end{tabular}
  \caption{$\h$-invariant tensors and the corresponding gauge-invariant $4$-forms}
  \label{tab:invariant-tensors}
\end{table}

\subsection{Variations}
\label{sec:variations}

We now vary the gauge-invariant $4$-forms in
Table~\ref{tab:invariant-tensors} with respect to $\omega^{ab}$,
$\psi^a$, $\theta^a$ and $\xi$.  We shall find that upon
application of the Bianchi identities \eqref{eq:bianchi-adsc-split}
that several of the terms vary into exact forms and hence do not
contribute to the Euler--Lagrange equations in the absence of a
boundary.  The variations of the curvatures can be found from the
expressions in \eqref{eq:variations} (or by direct variation of the
curvatures in \eqref{eq:curvatures}).  One finds the following
\begin{equation}
  \label{eq:variations-adsc}
  \begin{split}
    \delta \Omega^{ab} &= d^\nabla \delta\omega^{ab} + \sigma \left( \theta^a \wedge \delta\theta^b - \theta^b \wedge \delta\theta^a  \right)\\
    \delta \Psi^a &= d^\nabla \delta\psi^a + \delta\omega^a{}_b \wedge  \psi^b + \sigma \left( \xi \wedge \delta\theta^a - \theta^a \wedge \delta\xi\right) \\
    \delta \Theta^a &= d^\nabla \delta\theta^a + \delta\omega^a{}_b \wedge \theta^b \\
    \delta \Xi &= d\delta\xi + \psi_a \wedge \delta\theta^a - \theta^a \wedge \delta\psi_a,
  \end{split}
\end{equation}
where the last line is for completeness, since $\Xi$ does not appear
in the lagrangian.  Since $\xi$ must therefore appear linearly,
we may think of it as a Lagrange multiplier.

We now vary the gauge-invariant $4$-forms in turn, starting with those
of degree $0$ relative to the Lie algebra grading.  We do one such
calculation in detail to illustrate the method and list the results of
the other calculations.  Let us consider the variation of $\tfrac14
\Omega^{ab}\wedge \Omega_{ab}$:
\begin{equation}
  \label{eq:variation-example-calc}
  \begin{split}
    \delta(\tfrac14 \Omega^{ab} \wedge \Omega_{ab}) &= \tfrac12 \Omega_{ab} \wedge \delta \Omega^{ab}\\
    &= \tfrac12 \Omega_{ab} \wedge \left(d^\nabla \delta\omega^{ab} + 2 \sigma \theta^a \wedge \delta\theta^b \right)\\
    &= d(\tfrac12 \Omega_{ab} \wedge \delta\omega^{ab}) - \tfrac12  d^\nabla \Omega_{ab} \wedge \delta\omega^{ab} + \sigma \Omega_{ab} \wedge  \theta^a \wedge \delta\theta^b\\ 
    &= d(\tfrac12 \Omega_{ab} \wedge \delta\omega^{ab}) + \sigma \theta_a \wedge \Theta_b \wedge \delta\omega^{ab} + \sigma \Omega_{ab} \wedge  \theta^a \wedge \delta\theta^b,
  \end{split}
\end{equation}
where in the second line we have used
\begin{equation}
  d(\Omega_{ab} \wedge \delta\omega^{ab}) = d^\nabla\Omega_{ab} \wedge \delta\omega^{ab} + \Omega_{ab} \wedge d^\nabla \delta\omega^{ab}
\end{equation}
and in the third line we have used the first Bianchi identity of \eqref{eq:bianchi-adsc-split}.

In the same way we find the variations of all the gauge-invariant
$4$-forms in Table~\ref{tab:invariant-tensors} of zero Lie algebra degree:
\begin{equation}
  \label{eq:variations-4-forms-zero-deg}
  \begin{split}
    \delta(\tfrac14 \Omega^{ab} \wedge \Omega_{ab}) &= d(\tfrac12 \Omega_{ab} \wedge \delta\omega^{ab}) +\sigma\left( \theta_a \wedge \Theta_b \wedge \delta\omega^{ab} + \Omega_{ab} \wedge  \theta^a \wedge \delta\theta^b\right)\\
    \delta(\tfrac12 \Theta^a \wedge \Theta_a) &= d(\Theta_a \wedge \delta\theta^a) + \Omega_{ab} \wedge \theta^a \wedge \delta\theta^b - \Theta_a \wedge \theta_b \wedge \delta\omega^{ab}\\
    \delta(\tfrac12 \epsilon_{abc}\Theta^a \wedge\Omega^{bc}) &= d\left( \tfrac12 \epsilon_{abc} ( \delta\theta^a \wedge \Omega^{bc} + \Theta^a \wedge \delta\omega^{bc}) \right)\\
    \delta(\tfrac12 \theta^a \wedge \theta^b \wedge \Omega_{ab}) &= d (\tfrac12 \theta_a \wedge \theta_b \wedge\delta\omega^{ab}) + \Omega_{ab} \wedge \theta^a \wedge \delta\theta^b - \Theta_a \wedge \theta_b  \wedge \delta\omega^{ab}\\
    \delta(\tfrac12 \epsilon_{abc} \theta^a \wedge \theta^b \wedge \Theta^c) &= d(\tfrac12 \epsilon_{abc} \theta^a \wedge \theta^b \wedge \delta\theta^c),
  \end{split}
\end{equation}
In the middle variation, we have used that
\begin{equation*}
  \tfrac12 \omega^{ad}\wedge (\epsilon_{abd}\theta_c -
  \epsilon_{abc}\theta_d)\wedge\delta \omega^{bc} = 0
\end{equation*}
and elsewhere we have used that any totally antisymmetric quantity
$A^{abc}$ in three dimensions can be written as
$A^{abc} = \epsilon^{abc}A^{123}$.  We see from this that the
following $4$-forms vary into exact forms:
$\tfrac12 \epsilon_{abc}\Theta^a \wedge\Omega^{bc}$,
$\tfrac12 \epsilon_{abc} \theta^a \wedge \theta^b \wedge \Theta^c$,
$\tfrac14 \Omega^{ab} \wedge \Omega_{ab} - \tfrac12 \sigma \Theta^a \wedge
\Theta_a$,
$\tfrac12 \Theta^a \wedge \Theta_a - \tfrac12 \theta^a \wedge \theta^b 
\wedge \Omega_{ab}$
and
$\tfrac14 \Omega^{ab} \wedge \Omega_{ab} - \tfrac12 \sigma \theta^a \wedge
\theta^b \wedge \Omega_{ab}$, so that as far as the Euler--Lagrange
equations are concerned we may take one of
$\tfrac14 \Omega^{ab} \wedge \Omega_{ab}$,
$\tfrac12 \Theta^a \wedge \Theta_a$ or
$\tfrac12 \theta^a \wedge \theta^b \wedge \Omega_{ab}$.

Varying the remaining three gauge-invariant $4$-forms of nonzero Lie algebra
degree results in the following:
\begin{equation}
  \delta(\tfrac16 \epsilon_{abc}\theta^a\wedge \theta^b \wedge \theta^c \wedge \xi) = -\tfrac12 \epsilon_{abc}\xi \wedge \theta^a \wedge \theta^b \wedge \delta\theta^c + \tfrac16 \epsilon_{abc} \theta^a \wedge \theta^b \wedge \theta^c \wedge \delta \xi,
\end{equation}
\begin{multline}
  \delta(\tfrac12 \epsilon_{abc}\Psi^a \wedge \Omega^{bc}) =
  d\left(\tfrac12 \epsilon_{abc} (\delta\psi^a \wedge \Omega^{bc} +
    \Psi^a \wedge \delta\omega^{bc})\right) + \tfrac12 \sigma
  \epsilon_{abc} (\xi \wedge \Omega^{ab} + 2 \Psi^a \wedge \theta^b) \wedge \delta\theta^c\\
  - \tfrac12 \sigma \epsilon_{abc} \theta^a \wedge \Omega^{bc} \wedge
  \delta\xi + \sigma \epsilon_{abc} \theta^a \wedge \Theta^b
  \wedge \delta\psi^c  - \tfrac12 \sigma \epsilon_{abc}    (\theta^a \wedge \Xi - \xi \wedge \Theta^a)  \wedge \delta \omega^{bc}
\end{multline} 
and
\begin{multline}
  \delta(\tfrac12 \epsilon_{abc} (\theta^a \wedge \theta^b \wedge \Psi^c + \xi \wedge \theta^a \wedge \Omega^{bc})) = d\left( \tfrac12 \epsilon_{abc}(\theta^a \wedge \theta^b \wedge \delta \psi^c  + \xi \wedge \theta^a \wedge \delta\omega^{bc}) \right)\\
  + \tfrac12 \epsilon_{abc} (\xi \wedge \Theta^a - \Xi \wedge \theta^a) \wedge \delta\omega^{bc} + \epsilon_{abc} \theta^a  \wedge \Theta^b \wedge \delta\psi^c
  - \tfrac12 \epsilon_{abc} \theta^a \wedge (\sigma \theta^b \wedge \theta^c + \Omega^{bc}) \wedge \delta\xi \\
  + \tfrac12 \epsilon_{abc} (\xi \wedge (\Omega^{ab} + 3\sigma \theta^a \wedge \theta^b) - 2 \theta^a\wedge \Psi^b) \wedge \delta\theta^c.
\end{multline}

We notice that the variation of the combination
\begin{equation}
  \tfrac12 \epsilon_{abc}\Psi^a \wedge \Omega^{bc} - \tfrac12 \sigma \epsilon_{abc} \left( \theta^a \wedge \theta^b \wedge \Psi^c +
    \xi \wedge \theta^a \wedge \Omega^{bc} \right) - \tfrac12 \sigma
  \epsilon_{abc} \theta^a \wedge \theta^b \wedge \theta^c \wedge \xi
\end{equation}
is an exact form and hence it is not necessary to include all three
terms in the lagrangian.

We have summarised these findings in Table~\ref{tab:4-forms-car-sum}.

\subsection{Euler--Lagrange equations}
\label{sec:euler-lagr-equat}

Let us introduce real parameters $\beta, \mu,\Lambda$ and
let us consider the following gauge-invariant lagrangian $4$-form:
\begin{equation}\label{eq:lagrangian}
  \eL =
  \tfrac{\mu}2\epsilon_{abc} \left( \theta^a \wedge \theta^b \wedge \Psi^c +
    \xi \wedge \theta^a \wedge \Omega^{bc} \right)
  +\tfrac\beta2 \theta^a \wedge \theta^b \wedge \Omega_{ab}
  + \tfrac{\Lambda}6 \epsilon_{abc}\theta^a \wedge \theta^b \wedge \theta^c \wedge \xi .
\end{equation}
The first term is the carrollian
analog of the Hilbert–Palatini term, the second can be understood as a
carrollian Holst term, while the third is a carrollian cosmological
constant term, as listed in \eqref{eq:relativistic-lagrangian}.

Its variation results in an expression of the form
\begin{equation}
  \delta \eL = E \wedge \delta\xi + F_c \wedge \delta\theta^c +
  G_c \wedge \delta\psi^c + \tfrac12 H_{bc} \wedge \delta
  \omega^{bc} + d\Upsilon,
\end{equation}
where
\begin{equation}
  \Upsilon = \tfrac\beta2 \theta^a \wedge \theta^b \wedge \delta
  \omega_{ab} + \tfrac{\mu}2 \epsilon_{abc} \left( \theta^a\wedge \theta^b\wedge\delta\psi^c + \xi \wedge \theta^a\wedge \delta \omega^{bc}\right)
\end{equation}
and the Euler--Lagrange equations are given by
\begin{equation}
  \label{eq:EL-eqns}
  \begin{split}
    E &= \tfrac\Lambda6 \epsilon_{abc} \theta^a \wedge \theta^b \wedge \theta^c - \tfrac\mu2 \epsilon_{abc} \theta^a \wedge \Omega^{bc} =  0\\
    F_c &= \beta \theta^a \wedge \Omega_{ac} - \tfrac\Lambda2 \epsilon_{abc} \xi \wedge \theta^a \wedge \theta^b + \tfrac\mu2 \epsilon_{abc} \xi \wedge \Omega^{ab} - \mu \epsilon_{abc} \theta^a \wedge \Psi^b =0\\
    G_c &= \mu \epsilon_{abc} \theta^a \wedge \Theta^b = 0\\
    H_{bc} &= \beta \left( \theta_b \wedge \Theta_c - \theta_c \wedge \Theta_b  \right) - \mu \epsilon_{abc} \left( \theta^a \wedge \Xi - \xi \wedge \Theta^a \right) = 0.
  \end{split}
\end{equation}
Notice that the lagrangian does not depend explicitly on
$\sigma$.

Notice also that if $\mu = 0$, then the equation $E=0$ says that $\Lambda =
0$ since the $\theta^a$ generate a free graded commutative algebra.
This then requires $\beta \neq 0$ for a nontrivial lagrangian.  The
equations in this case reduce to only two:
\begin{equation}
    \theta^a \wedge \Omega_{ac} = 0 \qquad\text{and}\qquad \theta_b \wedge \Theta_c - \theta_c \wedge \Theta_b  = 0,
\end{equation}
which involve neither $\Xi$ nor $\Psi^a$, which remain unconstrained.
The first equation is solved by $\Omega_{ab} = \tfrac12 R_{abcd}
\theta^c \wedge \theta^d$ for $R_{abcd}$ the components of an
algebraic curvature tensor which is otherwise unconstrained.  The
second equation is solved by $\Theta^a = \tfrac12 \Theta^a_{bc}
\theta^b \wedge \theta^c$ where $\Theta^a_{ab}=0$.  These are
kinematical constraints and do not give rise to interesting dynamics.
We will therefore assume from now on that $\mu \neq 0$.\footnote{One 
might wonder whether the theory with $\mu = 0$ is the ``electric
theory'' of~\cite{Henneaux:1979vn,Hansen:2021fxi,Henneaux:2021yzg},
which only involves the carrollian intrinsic torsion. However, a more careful analysis reveals that this is not, in fact,
the electric theory. We comment on this in Section
\ref{sec:conclusion}.}

Since $\mu \neq 0$, the equations $G_c = 0$ and $H_{bc} = 0$ reduce to
\begin{subequations}
    \begin{align}
      \theta^{[a} \wedge \Theta^{b]} &= 0\label{eq:torsion-free-1}\\
      \theta^a \wedge \Xi - \xi \wedge \Theta^a &= 0,\label{eq:torsion-free-2}
  \end{align}
 where the first follows from $0 = \epsilon^{abc}G_c$.  
\end{subequations}

\begin{lemma}\label{lem:torsion-free-adsc}
  Equations~\eqref{eq:torsion-free-1} and \eqref{eq:torsion-free-2}
  imply that $\Xi = 0$ and $\Theta^a = 0$.
\end{lemma}

\begin{proof}
  We expand $\Xi$ and $\Theta^a$ in terms of the coframe
  $\theta^a,\xi$ as follows:
  \begin{equation*}
    \begin{split}
      \Theta^a &= \tfrac12 \Theta^a{}_{bc} \theta^b \wedge \theta^c +
      \Theta^a{}_b \theta^b \wedge \xi\\
      \Xi &= \tfrac12 \Xi_{ab} \theta^a \wedge \theta^b + \Xi_a
      \theta^a \wedge \xi.
    \end{split}
  \end{equation*}
  Equation~\eqref{eq:torsion-free-1} is equivalent to
  \begin{equation*}
    \begin{split}
      0 &= \epsilon_{abc} \theta^b \wedge \Theta^c\\
      &= \tfrac12 \epsilon_{abc} \Theta^c_{de} \theta^b \wedge
      \theta^d \wedge \theta^e + \epsilon_{abc} \Theta^c{}_d \theta^b
      \wedge \theta^d \wedge \xi.
    \end{split}
  \end{equation*}
  Both terms are independent and hence must vanish separately.  The
  first term gives the equation
  \begin{equation*}
    \epsilon^{bde} \epsilon_{abc} \Theta^c{}_{de} = 0
  \end{equation*}
  which implies $\Theta^d{}_{de} = 0$ after using the standard identity
  \begin{equation}
    \label{eq:faulkner-3d}
    \epsilon^{abc}\epsilon_{def} = 6 \delta^{[a}_d \delta^b_e \delta^{c]}_f.
  \end{equation}
  The second term gives the equation
  \begin{equation*}
    \epsilon^{bde}\epsilon_{abc} \Theta^c{}_d = 0 \implies
    \Theta^c{}_d = \Theta^b{}_b\delta^c_d,
  \end{equation*}
  again using equation~\eqref{eq:faulkner-3d}.  Tracing we see that
  $\Theta^b{}_b = 0$ and hence that $\Theta^c{}_d = 0$.

  Equation~\eqref{eq:torsion-free-2} is equivalent to
  \begin{equation*}
    0 = \tfrac12 \Xi_{cd} \theta^a \wedge \theta^c \wedge \theta^d +
    \Xi_c \theta^a \wedge \theta^c \wedge \xi - \tfrac12
    \Theta^a{}_{cd} \xi \wedge \theta^c \wedge \theta^d.
  \end{equation*}
  The $\theta^3$ term is equivalent to $\epsilon^{acd} \Xi_{cd} = 0$,
  whose only solution is $\Xi_{cd} =0$.  The $\theta^2\xi$ terms
  become
  \begin{equation*}
    \Xi_d \delta^a_c - \Xi_c \delta^a_d = \tfrac12 \Theta^a{}_{cd}.
  \end{equation*}
  Tracing and using that $\Theta^d{}_{de} = 0$, we see that $\Xi_d =
  0$ and hence that $\Theta^a{}_{cd} = 0$, thus proving the lemma.
\end{proof}

This lemma is analogous to the well-known fact that in the Palatini
formalism of General Relativity, the Euler--Lagrange equation of the
connection sets the torsion to zero.  Whereas this allows us to solve
for the connection in terms of the metric and essentially work in a
second-order formalism, this is not the case here.  Indeed, the Cartan
geometry in question is a carrollian $G$-structure and as shown in
\cite[Lemma~7]{Figueroa-OFarrill:2020gpr}, the torsion-free condition
does not determine the connection uniquely, but only up to the kernel
of the Spencer map, which for carrollian structures corresponds to
symmetric tensors of the form $S_{ab} \theta^a \theta^b$.

Since $\Xi$ and $\Theta^a$ vanish, their Bianchi identities then imply
that $\Omega^a{}_b \wedge \theta^b = 0$ and $\theta^a \wedge \Psi_a =
0$.  Let us expand $\Omega_{ab}$ and $\Psi_a$ in terms of the coframe
$\theta^a,\xi$ as follows
\begin{equation}
\label{eq:curvatures-expansion}
  \begin{split}
    \Omega_{ab} &= - \tfrac12 R_{abcd} \theta^c \wedge \theta^d +  R_{abc} \theta^c \wedge \xi\\
    \Psi_a &= \tfrac12 \Psi_{abc} \theta^b \wedge \theta^c + \Psi_{ab} \theta^b \wedge \xi.
  \end{split}
\end{equation}
The Bianchi identity $\Psi^a \wedge \theta_a = 0$ becomes
\begin{equation}
\label{eq:condition-1}
  \Psi^{[abc]} = 0 \qquad\text{and}\qquad  \Psi^{[ab]} = 0.
\end{equation}
The Bianchi identity $\Omega^a{}_b \wedge \theta^b = 0$ becomes
\begin{equation}
\label{eq:condition-2}
  R_{a[bcd]}=0 \qquad\text{and}\qquad  R^{a[bc]} = 0.
\end{equation}
The latter equation says that $R^{abc}=0$.  Indeed,
that equation together with $R^{(ab)c}=0$ say that
\begin{equation}
\label{eq:spatial-omega}
  R^{abc} = - R^{bac} = - R^{bca} = R^{cba} = R^{cab} = - R^{acb} =
  -R^{abc}.
\end{equation}
The former equation says that $R_{abcd}$ is an algebraic curvature
tensor.  Let us introduce the algebraic Ricci tensor
$R_{ab} = \eta^{cd} R_{acdb}$ and the Ricci scalar
$R = \eta^{ab} R_{ab}$.  In three dimensions (due to the absence of
a Weyl tensor) an algebraic curvature tensor $R_{abcd}$ has the
following Ricci decomposition
\begin{equation}\label{eq:ricci-dec}
  R_{abcd} = R_{ad} \delta_{bc} - R_{ac} \delta_{bd} - R_{bd}\delta_{ac} + R_{bc} \delta_{ad} + \tfrac12 R  (\delta_{ac}\delta_{bd} - \delta_{bc}\delta_{ad}),
\end{equation}
so that it is completely determined by its Ricci tensor.

The remaining Euler--Lagrange equations are
\begin{subequations}
  \begin{align}
    0 &= \tfrac\Lambda6 \epsilon_{abc} \theta^a \wedge \theta^b \wedge  \theta^c - \tfrac\mu2 \epsilon_{abc} \theta^a \wedge \Omega^{bc}\label{eq:eins-1}\\
    0 &= - \tfrac\Lambda2 \epsilon_{abc} \theta^a \wedge \theta^b  \wedge \xi + \mu \epsilon_{abc} \Psi^a \wedge \theta^b + \tfrac\mu2 \epsilon_{abc} \xi \wedge \Omega^{ab}\label{eq:eins-2}.
  \end{align}
\end{subequations}

\begin{lemma}
\label{lemma:Einstein-like-equation}
  Equation~\eqref{eq:eins-1} is equivalent to $R =
  \tfrac{2\Lambda}\mu$ and equation~\eqref{eq:eins-2} is equivalent to
  \begin{equation}
    \Psi^a{}_{ab} =0, \qquad\text{and}\qquad \Psi_{ab} = R_{ab} - \tfrac\Lambda\mu \delta_{ab}.
  \end{equation}
\end{lemma}

\begin{proof}
  Substitute $\Omega_{ab} = - \tfrac12 R_{abcd} \theta^c \wedge
  \theta^d$ into equation~\eqref{eq:eins-1} and simplify to obtain
  \begin{equation*}
    -\tfrac\mu4 \epsilon_{abc} R^{bc}{}_{de} \theta^a \wedge \theta^d
    \wedge \theta^e = \tfrac\Lambda6 \epsilon_{ade} \theta^a \wedge
    \theta^d \wedge \theta^e \implies -\tfrac\mu4
    \epsilon^{ade}\epsilon_{abc} R^{bc}{}_{de} = \tfrac\Lambda6 \epsilon^{ade}\epsilon_{ade},
  \end{equation*}
  which is seen to be $R  = \tfrac{2\Lambda}\mu$ after using the
  identity \eqref{eq:faulkner-3d}.

  Substituting $\Omega_{ab} = - \tfrac12 R_{abcd} \theta^c \wedge
  \theta^d$ and $\Psi_a = \tfrac12 \Psi_{abc} \theta^b \wedge \theta^c
  + \Psi_{ab} \theta^b \wedge \xi$ into
  equation~\eqref{eq:eins-2}, it decomposes into two equations.  The
  $\theta^3$ term says
  \begin{equation*}
    \tfrac\mu2 \epsilon_{abc} \Psi^a{}_{de} \theta^b \wedge \theta^d
    \wedge \theta^e \implies \epsilon^{bde}\epsilon_{abc} \Psi^a{}_{de}=0.
  \end{equation*}
  Using the identity~\eqref{eq:faulkner-3d}, this becomes simply the
  vanishing trace condition $\Psi^a{}_{ab}=0$.  The
  $\theta^2\xi$ terms result in the equation
  \begin{equation*}
    0 = -\tfrac\Lambda2
    \epsilon_{abc}\theta^a\wedge\theta^b\wedge\xi + \mu
    \epsilon_{abc}\Psi^a{}_d \theta^b \wedge \theta^d \wedge
    \xi - \tfrac\mu4 \epsilon_{abc} R^{ab}{}_{de} \theta^d
    \wedge \theta^e \wedge \xi,
  \end{equation*}
  or equivalently,
  \begin{equation*}
    0 = -\tfrac\Lambda2 \epsilon^{abf} \epsilon_{abc} + \mu
    \epsilon^{bdf} \epsilon_{abc}\Psi^a{}_d - \tfrac\mu4
    \epsilon^{def} \epsilon_{abc} R^{ab}{}_{de}.
  \end{equation*}
  Repeated use of the identity~\eqref{eq:faulkner-3d} and the fact
  that $R = \tfrac{2\Lambda}\mu$ and simplifying results in
  \begin{equation*}
    \Psi^a{}_a \delta^f_c - \Psi^f{}_c + R^{fe}{}_{ec} = 0.
  \end{equation*}
  Tracing and using again that $R = \tfrac{2\Lambda}\mu$, we see that
  $\Psi^a{}_a = -\tfrac\Lambda\mu$, and inserting back into the
  equation we find (after some relabelling) that
  \begin{equation*}
    \Psi_{ab} = R_{ab} - \tfrac\Lambda\mu \delta_{ab}.
  \end{equation*}
\end{proof}

The condition $\Psi^a{}_{ab} = 0$ together with $\Psi^{[abc]} = 0$ say that we can write
\begin{equation}
  \Psi^{abc} = S^a{}_d \epsilon^{bcd},
\end{equation}
for some traceless symmetric tensor $S_{ab}$.

We may summarise these results as follows, where we have set $\mu =1$
without loss of generality.

\begin{proposition}\label{prop:solution-adsc}
  The solution of the Euler--Lagrange equations for the lagrangian
  \begin{equation*}
   \eL =
  \tfrac{1}2\epsilon_{abc} \left( \theta^a \wedge \theta^b \wedge \Psi^c +
    \xi \wedge \theta^a \wedge \Omega^{bc} \right)
  +\tfrac\beta2 \theta^a \wedge \theta^b \wedge \Omega_{ab}
  + \tfrac{\Lambda}6 \epsilon_{abc}\theta^a \wedge \theta^b \wedge \theta^c \wedge \xi
  \end{equation*}
  are such that $\Xi = \Theta^a = 0$ and
  \begin{equation}
  \label{eq:eom-for-adsc}
    \begin{split}
      \Omega_{ab} &= R_{bc} \theta_a \wedge \theta^c - R_{ac} \theta_b \wedge \theta^c - 2\Lambda \theta_a \wedge \theta_b\\
      \Psi_a &= \tfrac12 \epsilon^{bcd} S_{ab} \theta_c \wedge \theta_d + R_{ab} \theta^b \wedge \xi - \Lambda \theta_a \wedge \xi,
      \end{split}
  \end{equation}
  where $S_{ab}$ is symmetric and traceless.
\end{proposition}

\section{The case in-between: Carroll}
\label{sec:carroll-gauging}

As we remarked above, setting $\sigma=0$
in~\eqref{eq:dsc-adsc-combined-brackets} gives rise to the Carroll
algebra. Since the homogeneous spaces of $\zdSC$, $\zAdSC$ and Carroll
are model mutants, that is to say, the stabiliser $\h$ is the same in
both cases and the kinematical Lie algebras are isomorphic as
$\h$-modules, the gauge-invariant $4$-forms are still given by those
listed in Table~\ref{tab:invariant-tensors}. Only quantities that
depend on the details of $\g$ itself are modified. This means that the
explicit expressions for the curvatures change, and consequently their
variations, and the Bianchi identities change as well. For Carroll, we
thus have \begin{equation}
  \label{eq:curvatures-carroll}
  \begin{split}
    \Omega^{ab} &= d\omega^{ab} + \omega^a{}_c \wedge \omega^{cb} \\
    \Psi^a &= d\psi^a + \omega^a{}_b \wedge \psi^b \\
    \Theta^a &= d\theta^a + \omega^a{}_b \wedge \theta^b \\
    \Xi &= d\xi + \psi^a \wedge \theta_a,
  \end{split}
\end{equation}
obtained by setting $\sigma=0$ in \eqref{eq:curvatures}.
Setting $\sigma=0$ in the derivation of \eqref{eq:lagrangian} we arrive at the same lagrangian
\begin{equation}
\label{eq:Carroll-lagrangian}
\eL=
 \tfrac{\mu}{2}\epsilon_{abc}\left( \theta^a\wedge\theta^b\wedge\Psi^c + \xi\wedge\theta^a\wedge \Omega^{bc} \right)
+ \tfrac{\beta}{2}\theta^a\wedge \theta^b\wedge \Omega_{ab}
+ \tfrac{\Lambda}{6}\epsilon_{abc}\theta^a\wedge\theta^b\wedge\theta^c\wedge\xi.
\end{equation}
and the same Euler--Lagrange equations where the curvature components
are now given by \eqref{eq:curvatures-carroll}. Note, however, that
boundary terms varying into total derivatives slightly change. We have
summarised these terms in Table \ref{tab:4-forms-car-sum}.

We can understand the Carroll lagrangian~\eqref{eq:Carroll-lagrangian} as the
``ultra-relativistic'' limit of the Einstein--Palatini--Holst
lagrangian~\eqref{eq:relativistic-lagrangian}.  To this end, consider
the Poincaré algebra~\eqref{eq:poincare} in the kinematical basis
$(L_{ab},B_a,P_a,H)$, where $a=1,2,3$ and
\begin{equation}
    B_a = L_{0a}\qquad\text{and}\qquad H = P_0.
\end{equation}
We can then write the Cartan connection and its associated curvature as
\begin{equation}
\begin{split}
A &= \tfrac12 \omega^{ab} L_{ab} + \psi^a B_a + \theta^a P_a + \xi H\\
F &= \tfrac12 \Omega^{ab} L_{ab} + \Psi^a B_a + \Theta^a P_a + \Xi H,
\end{split}
\end{equation}
where $\theta^m = (\xi,\theta^a)$,
$\omega^{mn} = (\psi^a,\omega^{ab})$, $\Theta^m = (\Xi,\Theta^a)$ and
$\Omega^{mn} = (\Psi^a,\Omega^{ab})$. Note that these expressions are
formally identical to~\eqref{eq:cartan-connection}
and~\eqref{eq:cartan-curvature}, although at this stage all we have
done is recast the Cartan connection and curvature of Minkowski. Now,
the ultra-relativistic limit corresponds to setting the speed of light
to zero, so to be able to take this limit we must introduce
appropriate factors of $c$ in the Minkowski metric, which consequently
becomes
\begin{equation}
    \eta_{mn} = \text{diag}(-c^2,1,1,1).
\end{equation}
In the new basis and with factors of $c$ restored, the Poincar\'e algebra~\eqref{eq:poincare} becomes
\begin{equation}
\begin{split}
    [L_{ab},L_{cd}] &= \delta_{bc}L_{ad} - \delta_{ac}L_{bd} - \delta_{bd}L_{ac} + \delta_{ad}L_{bc}\\
    [L_{ab},B_c] &= \delta_{bc}B_a - \delta_{ac}B_b\\
    [L_{ab},P_c] &= \delta_{bc}P_a - \delta_{ac}P_b\\
    [B_a,P_b] &= \delta_{ab}H\\
    [B_a,B_b] &= c^2L_{ab}\\
    [B_a,H] &= c^2P_a.
\end{split}
\end{equation}
Since $F = dA + \tfrac12 [A,A]$, the curvature components pick up factors of $c$
\begin{equation}
    \begin{split}
        \Omega^{ab} &= d^\nabla\omega^{ab} + c^2 \psi^a\wedge \psi^b\\
        \Psi^a &= d^\nabla \psi^a\\
        \Theta^a &= d^\nabla \theta^a + c^2 \psi^a\wedge \xi\\
        \Xi &= d\xi + \psi^a\wedge \theta_a,
    \end{split}
\end{equation}
which agree with~\eqref{eq:curvatures-carroll} when $c=0$. The gauge
invariant $4$-forms are those listed in Table~\ref{tab:4-forms-mink}.
The ones that do not vary into exact forms may now be written as
\begin{equation}
    \begin{split}
        \tfrac14\epsilon_{mnpq}\theta^m\wedge \theta^n\wedge \theta^p\wedge \theta^q &= \tfrac12 \epsilon_{abc}\xi\wedge\theta^a\wedge\Omega^{bc} + \tfrac12 \epsilon_{abc}\theta^a\wedge \theta^b\wedge\Psi^c\\
        \tfrac12 \theta^m\wedge \theta^n\wedge\Omega_{mn} &= \tfrac12 \eta_{mp}\eta_{nq} \theta^m\wedge \theta^n\wedge\Omega^{pq} = \tfrac12 \theta^a\wedge\theta^b \wedge \Omega_{ab} - c^2 \xi\wedge \theta^a\wedge \Psi_a\\
        \tfrac{1}{4!}\epsilon_{mnpq}\theta^m\wedge \theta^n\wedge \theta^p\wedge \theta^q &= -\tfrac{1}{3!}\epsilon_{abc} \theta^a\wedge\theta^b\wedge\theta^c\wedge\xi,
    \end{split}
\end{equation}
where we used that $\epsilon_{0abc} = \epsilon_{abc}$. This means that we can write the Minkowski lagrangian as
\begin{equation}
    \begin{split}
          \eL &= \tfrac{\mu}{4} \epsilon_{mnpq} \theta^m \wedge \theta^n \wedge
  \Omega^{pq} + \tfrac\beta2 \theta^m \wedge \theta^n \wedge \Omega_{mn}
  - \tfrac{\Lambda}{4!} \epsilon_{mnpq} \theta^m \wedge 
  \theta^n \wedge \theta^p \wedge \theta^q\\
  &= \tfrac{\mu}{2} \epsilon_{abc}(\xi\wedge\theta^a\wedge\Omega^{bc} + \theta^a\wedge \theta^b\wedge\Psi^c) + \tfrac\beta2 \theta^a\wedge\theta^b \wedge \Omega_{ab} - c^2 \beta \xi\wedge \theta^a\wedge \Psi_a\\
  &\quad + \tfrac{\Lambda}{3!}\epsilon_{abc} \theta^a\wedge\theta^b\wedge\theta^c\wedge\xi.
    \end{split}
\end{equation}
Setting $c=0$ in this expression produces the lagrangian
\begin{equation}
  \label{eq:llimc}
    \eL_{c=0} = \tfrac{\mu}{2} \epsilon_{abc}(\xi\wedge\theta^a\wedge\Omega^{bc} + \theta^a\wedge \theta^b\wedge\Psi^c) + \tfrac\beta2 \theta^a\wedge\theta^b \wedge \Omega_{ab} + \tfrac{\Lambda}{3!}\epsilon_{abc} \theta^a\wedge\theta^b\wedge\theta^c\wedge\xi,
\end{equation}
which is identical to the Carroll lagrangian~\eqref{eq:Carroll-lagrangian}.

To match this with the result of~\cite{Bergshoeff:2017btm}, we
take~\eqref{eq:Carroll-lagrangian} and set $\mu = 1$ and
$\beta = \Lambda = 0$, leaving only the term that descends from the
Einstein--Hilbert--Palatini term
\begin{equation}
  \eL = \tfrac{1}{2}\epsilon_{abc}\left(
    \theta^a\wedge\theta^b\wedge\Psi^c + \xi\wedge\theta^a\wedge
    \Omega^{bc} \right).
\end{equation}
Then, write for, say, the first term the following
\begin{equation}
\label{eq:first-term}
    \tfrac{1}{2}\epsilon_{abc}
    \theta_\mu{^a}\theta_\nu{^b}\Psi_{\rho\sigma}{^c}dx^\mu\wedge
    dx^\nu\wedge dx^\rho\wedge dx^\sigma =
    \tfrac{1}{2e}\epsilon_{abc}\epsilon^{\mu\nu\rho\sigma}
    \theta_\mu{^a}\theta_\nu{^b}\Psi_{\rho\sigma}{^c}\vol,
\end{equation}
where we introduced local coordinates $\{x^\mu\}$ on the
four-dimensional manifold, and where $e =
\det(\xi_\mu,\theta_\mu{^a})$ and we defined the top form
\begin{equation}
  \label{eq:defdvol}
  \vol := \tfrac16 \epsilon_{abc} \theta^a \wedge \theta^b \wedge
  \theta^c \wedge \xi = \theta^1 \wedge \theta^2 \wedge \theta^3
  \wedge \xi \in \Omega^4(U).
\end{equation}
 More specifically, the combination
$(\xi_\mu,\theta_\mu{^a})$ forms an invertible matrix, and we have
that
\begin{equation}
  e = \tfrac{1}{3!}\epsilon^{\mu\nu\rho\sigma}\epsilon_{abc}\xi_\mu\theta_\nu{^a}\theta_\rho{^b}\theta_\sigma{^c}.
\end{equation}
As above, we denote the inverse of $(\xi_\mu,\theta_\mu{^a})$ by
$(\kappa^\mu,e^\mu{}_a)$, and inserting the completeness
relation~\eqref{eq:frame-coframe} twice on the right-hand side
of~\eqref{eq:first-term}, we get
\begin{equation}
    \begin{split}
        \tfrac{1}{2e}\epsilon_{abc}\epsilon^{\mu\nu\rho\sigma} \theta_\mu{^a}\theta_\nu{^b}\Psi_{\rho\sigma}{^c}\vol &= \tfrac{1}{e}\epsilon_{abc}\epsilon^{\mu\nu\rho\sigma}\xi_\rho \theta_\mu{^a}\theta_\nu{^b}\theta_\sigma{^d}\kappa^\lambda e^\kappa{_d}\Psi_{\lambda\kappa}{^c}\vol\\
        &=\tfrac{1}{e}\epsilon_{abc}\epsilon^{\mu\nu\rho\sigma}\xi_\rho \theta_\mu{^a}\theta_\nu{^b}\theta_\sigma{^c}\kappa^\lambda e^\kappa{_d}\Psi_{\lambda\kappa}{^d}\vol\\
        &= \kappa^\mu e^\nu{_a}\Psi_{\mu\nu}{^a}\vol,
    \end{split}
\end{equation}
where in going from the first to the second line, we used that the
Levi--Civita symbol forces the $d$-index on $\theta_\sigma{^d}$ to be
equal to $c$. A similar calculation shows that
\begin{equation}
\label{eq:first-order-NotStarWars}
    \tfrac{1}{2}\epsilon_{abc}\xi\wedge\theta^a\wedge \Omega^{bc} =\tfrac12 e^\mu{_a}e^\nu{_b}\Omega_{\mu\nu}{}^{ab}\vol,
\end{equation}
and so we recover the lagrangian of~\cite{Bergshoeff:2017btm} in the
first-order formulation:
\begin{equation}
    \eL = \tfrac12\left(e^\mu{_a}e^\nu{_b}\Omega_{\mu\nu}{}^{ab}+2\kappa^\mu e^\nu{_a}\Psi_{\mu\nu}{^a} \right)\vol.
\end{equation}
To get the lagrangian in a second order formulation, we make use of
the results we obtain in the next section below. There, we will find that
\begin{equation}
    \begin{split}
        \Omega_{\mu\nu}{}^{ab} &= e^{\lambda a}\theta_{\rho}{}^bR^{\nabla}_{\mu\nu\lambda}{^\rho}\\
        \Psi_{\mu\nu}{^a} &= -\xi_\rho e^{\lambda a}R^\nabla_{\mu\nu\lambda}{^\rho},\\
    \end{split}
\end{equation}
where $R^{\nabla}_{\mu\nu\lambda}{^\rho}$ are the components of the Riemann tensor
of a certain affine connection $\nabla$. Plugging this into the first-order
lagrangian~\eqref{eq:first-order-NotStarWars}, we get
\begin{equation}
\label{eq:carroll-lagrangian}
    \eL = \tfrac12 \gamma^{\mu\nu} \left( R^\nabla_{\mu\nu} +
      \kappa^\rho \xi_\lambda R^\nabla_{\mu\rho\nu}{^\lambda} \right) \vol,
\end{equation}
in agreement with~\cite{Bergshoeff:2017btm}.

\section{Geometric interpretation}
\label{sec:geom-interpr}

We can recast our results in terms of the carrollian structure on the
manifold $M$.  We work in a Cartan gauge $(U,A)$ where $(U,x^\mu)$ is
also a chart with local coordinates $x^\mu = (x^0,x^1,x^2,x^3)$,
relative to which the Cartan connection $A \in \Omega^1(U;\g)$ given
in \eqref{eq:cartan-connection} has components
\begin{equation}
  \label{eq:local-connection}
  A_\mu = \tfrac12\omega_\mu{^{ab}}L_{ab} + \psi_\mu{^a}B_a +
 \theta_\mu{^a}P_a + \xi_\mu H.
\end{equation}
Since $\zAdSCp$ and $\zC$ are reductive Klein geometries, the
$\m$-component $\theta^a P_a + \xi H$ of $A$ defines a coframe: an
isomorphism $T_pM \to \m$, sending $\d_\mu \mapsto \theta_\mu^a P_a +
\xi_\mu H$.  Let $(\kappa, e_a)$ denote the canonically dual frame.  The
carrollian structure is given by $(\kappa,h)$, where $h = \delta_{ab}
\theta^a \theta^b$.  Relative to local coordinates,
\begin{equation}
  \kappa = \kappa^\mu \d_\mu, \quad \xi = \xi_\mu dx^\mu, \quad \theta^a =
  \theta^a_\mu dx^\mu \quad\text{and}\quad e_a = e_a^\mu \d_\mu,
\end{equation}
satisfying the following completeness relations:
\begin{equation}\label{eq:frame-coframe}
  \xi_\mu \kappa^\mu = 1, \quad \theta^a_\mu \kappa^\mu = 0, \quad \xi_\mu
  e_a^\mu = 0, \quad \theta^a_\mu e_b^\mu = \delta^a_b
  \quad\text{and}\quad \xi_\mu \kappa^\nu + \theta^a_\mu e_a^\nu = \delta_\mu^\nu.
\end{equation}
It is convenient to define $\gamma \in \Gamma(\odot^2 TU)$ via
$\gamma = \delta^{ab} e_a e_b$ with components $\gamma^{\mu\nu} =
\delta^{ab} e_a^\mu e_b^\nu$.  It follows from the completeness
relations that
\begin{equation}\label{eq:completeness}
  \gamma^{\mu\rho}h_{\rho\nu} + \kappa^\mu \xi_\nu = \delta^\mu_\nu,
\end{equation}
where $h_{\mu\nu} = \delta_{ab}\theta^a_\mu \theta^b_\nu$ are the
components of $h$.

The $\h$-component $A^\h$ of $A$ defines an Ehresmann connection on
the principal $H$-bundle $P \to M$ of the Cartan geometry.  This
connection induces a Koszul connection on the so-called fake tangent
bundle $P\times_H \m \to M$, which is the vector bundle associated to
$P$ via the $H$-representation $\m$.  Locally, sections of
$P\times_H \m$ are $\m$-valued functions on $U$.  The coframe
$(\xi,\theta^a)$ gives a bundle isomorphism $TM \to P\times_H \m$
sending $\d_\mu \mapsto \xi_\mu H + \theta^a_\mu P_a$ and allowing us
to transport the Koszul connection to an affine connection $\nabla$ on
$TM$.  This affine connection has connection coefficients
$\Gamma_{\mu\nu}^\rho$ defined by
\begin{equation}
  \nabla_\mu \d_\nu = \Gamma_{\mu\nu}^\rho \pd_{\rho}.
\end{equation}
They can be determined in terms of the Cartan connection via the
so-called Vierbein postulate, which says that we obtain the same
result if we differentiate $\d_\mu$ with the affine connection
$\nabla$ and then map to $P\times_H \m$ or first map to $P\times_H \m$
and differentiate with the $\h$-part of the Cartan connection.
Explicitly, in the former operation we obtain
\begin{equation}
  \begin{tikzcd}
    \partial_\nu \arrow[r,mapsto,"\nabla_\mu"] & \Gamma_{\mu\nu}^\rho \partial_\rho
    \arrow[r,mapsto,"\theta"] & \Gamma_{\mu\nu}^\rho \left(
      \xi_\rho H + \theta_\rho^a P_a \right)
  \end{tikzcd}
\end{equation}
whereas in the latter we obtain
\begin{equation}
  \begin{tikzcd}
    \partial_\nu \arrow[r,mapsto,"\theta"] & \xi_\nu H +
\theta^a_\nu P_a \arrow[r,mapsto,"\nabla_\mu"] & \partial_\mu \xi_\nu
H + \xi_\nu [A^\h_\mu,H] + \partial_\mu \theta^a_\nu P_a +
\theta^a_\nu [A^\h_\mu,P_a],
  \end{tikzcd}
\end{equation}
where $A^\h$ is the $\h$-component of the Cartan connection.
Equating the two expressions, we obtain
\begin{equation}
\Gamma_{\mu\nu}^\rho \left( \xi_\rho H + \theta_\rho^a P_a \right)   = \left(\partial_\mu\xi_\nu + \theta^a_\nu
\psi_{\mu\,a}\right) H + \left(\partial_\mu \theta^a_\nu + \omega_\mu{}^a{}_b
\theta^b_\nu\right) P_a,
\end{equation}
from where we read off the following identities:
\begin{equation}\label{eq:vierbein-postulate}
  \begin{split}
    \Gamma_{\mu\nu}^\rho \xi_\rho &= \partial_\mu\xi_\nu + \psi_{\mu\,a}\theta^a_\nu \\
    \Gamma_{\mu\nu}^\rho \theta^a_\rho &=\partial_\mu \theta^a_\nu  + \omega_\mu{}^a{}_b \theta^b_\nu.
  \end{split}
\end{equation}
The completeness relations~\eqref{eq:frame-coframe} say that
\begin{equation}
  \begin{split}
    \Gamma_{\mu\nu}^\rho &= \Gamma_{\mu\nu}^\sigma \delta_\sigma^\rho\\
    &= \Gamma_{\mu\nu}^\sigma \left( \xi_\sigma \kappa^\rho  + \theta^a_\sigma e_a^\rho \right)\\
    &= \Gamma_{\mu\nu}^\sigma \xi_\sigma \kappa^\rho  +  \Gamma_{\mu\nu}^\sigma \theta^a_\sigma e_a^\rho \\
  &= \left( \partial_\mu\xi_\nu + \theta^a_\nu \psi_{\mu\,a} \right)
  \kappa^\rho + \left( \partial_\mu \theta^a_\nu  + \omega_\mu{}^a{}_b \theta^b_\nu  \right) e_a^\rho\\
  \end{split}
\end{equation}
so that
\begin{equation}\label{eq:affine-conn}
  \Gamma_{\mu\nu}^\rho = \kappa^\rho \partial_\mu\xi_\nu + \kappa^\rho \theta^a_\nu
  \psi_{\mu\,a} + e_a^\rho \partial_\mu \theta^a_\nu  + e_a^\rho \omega_\mu{}^a{}_b \theta^b_\nu.
\end{equation}

Similarly, the curvature $F = dA + \tfrac12 [A,A]$ given in
\eqref{eq:cartan-curvature} has components
\begin{equation}
    \label{eq:local-curvature}
    F_{\mu\nu} = \tfrac12\Omega_{\mu\nu}{^{ab}} L_{ab} + \Psi_{\mu\nu}{^a}B_a + \Theta_{\mu\nu}{^a}P_a + \Xi_{\mu\nu}H,
\end{equation}
where the local forms of \eqref{eq:curvatures} read
\begin{equation}
    \begin{split}
        \Omega_{\mu\nu}{^{ab}} &= 2\D_{[\mu}\omega_{\nu]}{^{ab}} + 2\omega_{[\mu\vert}{^a}{_c}\omega_{\nu]}{^{cb}} + 2\sigma \theta^a_{[\mu}\theta^b_{\nu]}\\
        \Psi_{\mu\nu}{^a} &= 2\D_{[\mu}\psi_{\nu]}{^a} + 2\omega_{[\mu\vert }{^a}{_b}\psi_{\vert \nu]}{^b} + 2\sigma \xi_{[\mu}\theta_{\nu]}^a\\
        \Theta_{\mu\nu}{^a}&=2\D_{[\mu}\theta_{\nu]}^a + 2\omega_{[\mu\vert}{^a}{_b}\theta_{\vert \nu]}^b\\
        \Xi_{\mu\nu} &= 2\D_{[\mu}\xi_{\nu]} + 2\psi_{[\mu}{^a}\theta_{\nu]a}.
    \end{split}
\end{equation}
The last term in each of the two first lines --- which are also the
only terms involving the sign $\sigma$ --- constitute the only
difference compared to gauging the Carroll algebra; that is,
the algebra~\eqref{eq:dsc-adsc-combined-brackets} with $\sigma=0$
(see~\cite{Hartong:2015xda,Bergshoeff:2017btm,Guerrieri:2021cdz}).

In general, only for the Klein model in which $\m$ is abelian, does
the curvature of the affine connection derived from the $\h$-component
of the Cartan connection agree with the $\h$-component of the
curvature of the Cartan connection.  The Klein model with $\m$ abelian
is the one for which $\g$ is the Carroll algebra; that is,
the case where $\sigma = 0$.  Adorning with a hat the curvature with
$\sigma = 0$, we may thus write
\begin{equation}
    \label{eq:pure-carroll}
    \begin{split}
        \Omega_{\mu\nu}{^{ab}} &= \widehat\Omega_{\mu\nu}{^{ab}}+ 2\sigma\theta_{[\mu}{^a}\theta_{\nu]}{^b},\\
        \Psi_{\mu\nu}{^a} &= \widehat\Psi_{\mu\nu}{^a}+ 2\sigma\xi_{[\mu}\theta_{\nu]}{^a},
    \end{split}
\end{equation}
which makes explicit the dependence on $\sigma$.  As we will see
explicitly below, the hatted curvature is the curvature of the
$\h$-component of the Cartan connection.

The torsion of the affine connection~\eqref{eq:affine-conn} satisfies
\begin{equation}
   T_{\mu\nu}^\lambda :=  \Gamma_{\mu\nu}^\lambda - \Gamma_{\nu\mu}^\lambda = v^\lambda\Xi_{\mu\nu} + e^\lambda_a\Theta_{\mu\nu}{^a},
\end{equation}
and hence the equations of motion $\Theta^a = \Xi = 0$ imply that the
torsion of the affine connection vanishes.  In the literature,
conditions like $\Theta^a = \Xi = 0$ are sometimes called
curvature\footnote{Or, occasionally, \textit{conventional}
  constraints.} constraints and, in the absence of an action from
which they arise as equations of motion, are imposed by hand in order
that the affine connection can be solved for in terms of the soldering
form.

Since the affine connection $\nabla$ comes from an Ehresmann
connection on $P$, it is \emph{adapted} to the carrollian structure;
that is, the tensors $(\kappa,h)$ defining the carrollian structure are
parallel:
\begin{equation}
  \label{eq:adapted-connection}
  \nabla_\mu \kappa^\nu = 0 \qquad\text{and}\qquad \nabla_\mu h_{\nu\lambda} = 0.
\end{equation}

As shown in \cite{Bekaert:2015xua,Hartong:2015xda}, the most general
torsion-free connection adapted to the carrollian structure $(v,h)$
has coefficients given by
\begin{equation}
  \label{eq:torsion-free-adapted-connection}
  \Gamma^\lambda_{\mu\nu} = \kappa^\lambda\D_{(\mu}\xi_{\nu)} + \tfrac12 \gamma^{\lambda \rho}(\D_\mu
  h_{\rho \nu} +\D_\nu h_{\mu \rho} - \D_\rho h_{\mu\nu} ) -
  \kappa^\lambda\Sigma_{\mu\nu},
\end{equation}
where $\Sigma_{\mu\nu} = \Sigma_{(\mu\nu)}$ is an arbitrary symmetric
tensor subject to the condition
\begin{equation}\label{eq:v-parallel-contorsion}
  \kappa^\mu \Sigma_{\mu\nu} = 0,
\end{equation}
arising from the fact that $\kappa$ is parallel.  The connection
coefficients
\begin{equation}
  \label{eq:tilde-connection}
  \widetilde\Gamma^\lambda_{\mu\nu} := \kappa^\lambda\D_{(\mu}\xi_{\nu)} +
  \tfrac12 \gamma^{\lambda \rho}\left(\D_\mu h_{\rho \nu} +\D_\nu h_{\mu \rho} -
  \D_\rho h_{\mu\nu}\right)
\end{equation}
define another adapted connection $\widetilde\nabla$ and the difference
between the two adapted connections $\widetilde\nabla$ and $\nabla$ is the
$(1,2)$-tensor field with components $\kappa^\lambda \Sigma_{\mu\nu}$.
Notice that both $\Gamma_{\mu\nu}^\rho$ in
equation~\eqref{eq:affine-conn} and $\widetilde\Gamma_{\mu\nu}^\rho$ in
equation~\eqref{eq:tilde-connection} are given in terms of the
components of the Cartan connection, and hence so is their difference.
This means that we may solve for $\Sigma_{\mu\nu}$ in terms of the
Cartan connection.  Since the affine connection $\nabla$ is torsion free, we can write~\eqref{eq:affine-conn} as
\begin{equation}
    \Gamma^\lambda_{\mu\nu} = \kappa^\lambda \D_{(\mu}\xi_{\nu)} + \kappa^\lambda \theta_{(\mu}{^a}\psi_{\nu)a} + e^\lambda{_a}\D_{(\mu}\theta_{\nu)}{^a} + e^\lambda{_a}\omega_{(\mu\vert}{^a}{_b}\theta_{\vert \nu)}{^b}.
\end{equation}
On the other hand, the connection coefficients $\widetilde\Gamma^\lambda_{\mu\nu}$ of the adapted connection $\widetilde\nabla$ can be written in terms of the coframe $\theta_\mu{^a}$ by writing $h_{\mu\nu} = \delta_{ab}\theta_\mu{^a}\theta_\nu{^b} = \theta_\mu{^a}\theta_{\nu\, a}$, in which case we find that
\begin{equation}
    \begin{split}
         \widetilde\Gamma^\lambda_{\mu\nu} &= \kappa^\lambda \D_{(\mu}\xi_{\nu)}+ e^\lambda{_a}\D_{(\mu}\theta_{\nu)}{^a} + \tfrac12 \gamma^{\lambda\rho}\theta_{\nu\,a}(d\theta^a)_{\mu\rho} + \tfrac12 \gamma^{\lambda\rho}\theta_{\mu\,a}(d\theta^a)_{\nu\rho}\\
         &= \kappa^\lambda \D_{(\mu}\xi_{\nu)}+ e^\lambda{_a}\D_{(\mu}\theta_{\nu)}{^a} - \gamma^{\lambda\rho}\theta_{\nu\,a}\omega_{[\mu\vert}{^a}{_b}\theta_{\vert\rho]}{^b} - \gamma^{\lambda\rho}\theta_{\mu\,a}\omega_{[\nu\vert}{^a}{_b}\theta_{\vert\rho]}{^b}\\
         &= \kappa^\lambda \D_{(\mu}\xi_{\nu)}+ e^\lambda{_a}\D_{(\mu}\theta_{\nu)}{^a} + e^{\lambda}{_a}\theta_{(\nu\vert}{^b}\omega_{\vert\mu)}{^a}{_b},
    \end{split}
\end{equation}
where we used that $\Theta^a = 0$ in the second equality. Hence, we find that
\begin{equation}
\label{eq:sigma-tensor}
    \Gamma^{\lambda}_{\mu\nu}-\tilde{\Gamma}^{\lambda}_{\mu\nu}=\kappa^\lambda\theta^a_{(\mu}\psi_{\nu)a}.
\end{equation}

The Riemann tensor of an affine connection $\nabla$ is defined by
\begin{equation}
  R^\nabla(X,Y)Z = \nabla_{[X,Y]} Z - [\nabla_X,\nabla_Y] Z,
\end{equation}
whose components $R^\nabla_{\mu\nu\rho}{}^\sigma$ are defined by
\begin{equation}
  R^\nabla_{\mu\nu\rho}{}^\sigma \d_\sigma =
  R^\nabla(\d_\mu,\d_\nu)\d_\rho  = - [\nabla_\mu,\nabla_\nu] \d_\rho.
\end{equation}
In terms of the connection coefficients, they are given by
\begin{equation}\label{eq:Riemann-tensor-def}
  R^\nabla_{\mu\nu\lambda}{^\rho} = -\D_\mu\Gamma^\rho_{\nu\lambda} +
  \D_\nu \Gamma^\rho_{\mu\lambda} - \Gamma^\rho_{\mu
    \sigma}\Gamma^\sigma_{\nu\lambda} + \Gamma^\rho_{\nu \sigma}\Gamma^\sigma_{\mu\lambda}.
\end{equation}
Comparing with the Riemann tensor of $\widetilde\nabla$, which has a
similar expression but with $\widetilde\Gamma_{\mu\nu}^\rho$ replacing
$\Gamma_{\mu\nu}^\rho$, we have
\begin{equation}
\label{eq:split-of-Riemann}
    R^\nabla_{\mu\nu\lambda}{^\rho} =  R^{\widetilde\nabla}_{\mu\nu\lambda}{^\rho} + 2\kappa^\rho\widetilde\nabla_{[\mu}\Sigma_{\nu]\lambda}.
\end{equation}
Using the explicit expression \eqref{eq:affine-conn} for the
connection coefficients and using the Vierbein postulate
\eqref{eq:vierbein-postulate}, we may write the Riemann tensor of the
affine connection in terms of the curvatures that feature in
\eqref{eq:local-curvature}:
\begin{equation}
  \label{eq:curvature-relation}
    R^\nabla_{\mu\nu\lambda}{^\rho} = -\kappa^\rho\theta_{\lambda
      a}\Psi_{\mu\nu}{^a} + \theta_{\lambda
      a}e^\rho{_b}\Omega_{\mu\nu}{^{ab}} + 2\sigma
    \kappa^\rho\theta_{\lambda a}\xi_{[\mu}\theta_{\nu]}{^a} -
    2\sigma\theta_{\lambda
      a}e^\rho{_b}\theta_{[\mu}{^a}\theta_{\nu]}{^b},
\end{equation}
which, as advertised earlier, admits a simpler form in terms of the
pure (pseudo-)carrollian curvatures defined in \eqref{eq:pure-carroll}:
\begin{equation}
  \label{eq:curvature-relation-2}
  R^\nabla_{\mu\nu\lambda}{^\rho} = -\kappa^\rho\theta_{\lambda a}\widehat\Psi_{\mu\nu}{^a} + \theta_{\lambda a}e^\rho{_b}\widehat\Omega_{\mu\nu}{^{ab}}.
\end{equation}
This, together with the relations~\eqref{eq:frame-coframe}, then implies that
\begin{equation}
 \label{eq:expression-for-Omega-and-Psi}
     \begin{split}
         \widehat\Omega_{\mu \nu ab} &= e^{\lambda}{_a}\theta_{\rho b}R^\nabla_{\mu\nu\lambda}{^\rho}\\
         \widehat\Psi_{\mu\nu}{^a} &= -\xi_\rho e^{\lambda a} R^\nabla_{\mu\nu\lambda}{^\rho},
     \end{split}
\end{equation}
which we used in \eqref{eq:carroll-lagrangian} to write our lagrangian in a second-order
formulation.
 
The Ricci tensor of $\nabla$ is symmetric.  This follows from the
fact, proved in Appendix~\ref{sec:ricci-tensor-torsion}, that the
Ricci tensor for a torsion-free affine connection on an orientable
manifold is symmetric if and only if around every point there exists a
locally defined parallel volume form.  The assumption of orientability
is harmless by passing, if needed, to the orientation double cover.
This is $\h$-invariant and hence parallel relative to $\nabla$.
Indeed,
\begin{equation}
   \nabla_\mu \vol = \tfrac12 \epsilon_{abc} \nabla_\mu \theta^a
    \wedge \theta^b \wedge \theta^c \wedge \xi + \tfrac16 \epsilon_{abc} \theta^a \wedge \theta^b \wedge
  \theta^c \wedge \nabla_\mu\xi.
\end{equation}
From the Vierbein postulate \eqref{eq:vierbein-postulate}, we see that
\begin{equation}
  \nabla_\mu \theta^a = - \omega_\mu{}^a{}_d \theta^d
  \qquad\text{and}\qquad \nabla_\mu \xi = - \psi_{\mu\, d} \theta^d
\end{equation}
and inserting into $\nabla_\mu \vol$ we arrive at
\begin{equation}
   \nabla_\mu \vol = \tfrac12 \epsilon_{abc} \omega_\mu{}^a{}_d \theta^d
    \wedge \theta^b \wedge \theta^c \wedge \xi + \tfrac16 \epsilon_{abc} \theta^a \wedge \theta^b \wedge
  \theta^c \wedge \psi_{\mu\,d}\theta^d.
\end{equation}
The second term contains the wedge product of four $\theta$s, but
only three are available, hence it vanishes.  This leaves the first
term
\begin{equation}
    \tfrac12 \epsilon_{abc} \omega_\mu{}^a{}_d \theta^d \wedge
    \theta^b \wedge \theta^c \wedge \xi = \tfrac12 \epsilon_{abc} \epsilon^{dbc} \omega_\mu{}^a{}_d \vol = \delta^d_a \omega_\mu{}^a{}_d \vol,
\end{equation}
which is zero because $\omega_{ab} = -\omega_{ba}$.  In other words,
$\nabla_\mu \vol = 0$ and from Proposition~\ref{prop:symmetric-ricci},
it follows that the Ricci tensor of $\nabla$ is symmetric.

The Ricci tensor of $\widetilde\nabla$ is also symmetric.  Indeed, as
shown in Proposition~\ref{prop:contorsion-ricci-symmetric}, the Ricci
tensor of $\widetilde\nabla$ is symmetric if and only if the trace of
the contorsion is a closed one-form, but from
equation~\eqref{eq:torsion-free-adapted-connection} the contorsion
tensor has components $\kappa^\rho\Sigma_{\mu\nu}$ and hence its trace is
$\kappa^\mu\Sigma_{\mu\nu}$ which vanishes since $\kappa$ is parallel with
respect to any adapted connection (see
equation~\eqref{eq:v-parallel-contorsion}).

\section{Gauging the lightcone}
\label{sec:gauginglc}

In this section we apply the gauging procedure to the lightcone, by
which we mean the future lightcone with deleted apex as depicted in
Figure~\ref{fig:lightcone}. Cartan geometries modelled on the lightcone have previously been considered in~\cite{Palomo:2021lcg}.
\begin{figure}[h!]
    \centering
    \includegraphics[width=0.6\textwidth]{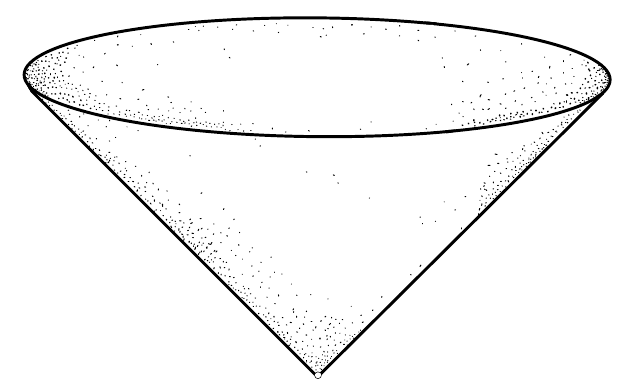}
    \caption{Future lightcone without apex}
    \label{fig:lightcone}
\end{figure}
The lightcone is a homogeneous space of the connected component of the
identity of the Lorentz group $\SO(4,1)_0$.  Its Lie algebra
$\g = \so(4,1)$ is spanned by $L_{ab},B_a, P_a, H$ with $a,b=1,2,3$
and subject to the nonzero brackets
\begin{equation}
  \label{eq:lightcone-g}
  \begin{split}
    [L_{ab},L_{cd}] &= \delta_{bc} L_{ad} - \delta_{ac} L_{bd} - \delta_{bd} L_{ac} + \delta_{ad} L_{bc}\\
    [L_{ab},B_c] &= \delta_{bc} B_a - \delta_{ac} B_b\\
    [L_{ab},P_c] &= \delta_{bc} P_a - \delta_{ac} P_b\\
    [H,B_a] &= B_a\\
    [H,P_a] &= -P_a\\
    [B_a, P_b] & = \delta_{ab} H + L_{ab}.
  \end{split}
\end{equation}
The Klein pair $(\g,\h)$ for the lightcone is such that $\h$ is the
span of $L_{ab}$ and $B_a$, which is isomorphic to the euclidean
algebra $\iso(3)$, in common with the other maximally symmetric
carrollian Klein pairs: $\zC$, $\zdSC$ and $\zAdSC$.  In contrast with
these other carrollian Klein pairs, the Klein pair of the lightcone is
\emph{not} reductive.  Letting $(\g',\h')$ be any one of these other
carrollian Klein pairs, we have that $\h \cong \h'$, but $\g'$ is not
isomorphic to $\g$ as a representation of $\h$.  Whereas $\zC$,
$\zdSC$ and $\zAdSC$ are all mutants of each other, the lightcone is
not related any of them by mutation.  The carrollian Cartan geometry
associated to the lightcone is therefore in principle different to the
one modelled on $\zC$, $\zdSC$ or $\zAdSC$.  In fact, the
carrollian structure of the lightcone is of totally umbilical type
(in the classification of \cite{Figueroa-OFarrill:2020gpr}), as shown in
\cite[§7.3]{Figueroa-OFarrill:2019sex}, and not of totally geodesic
type as that of $\zC$, $\zdSC$ or $\zAdSC$.

We let $\r = \left<L_{ab}\right> \cong \so(3)$, $\h =
\left<L_{ab},B_a\right> \cong \iso(3)$ and $\g =
\left<L_{ab},B_a,P_a,H\right> \cong \so(4,1)$, with the brackets given
explicitly by equation~\eqref{eq:lightcone-g}.

\subsection{The gauge fields}
\label{sec:gauge-fields-lc}

Let $M$ be a four-dimensional manifold and $U \subset M$ an open subset and let $A
\in \Omega^1(U,\g)$ be given relative to the basis of $\g$ by
\begin{equation}
  \label{eq:gauge-field-lc}
  A = \tfrac12 \omega^{ab}L_{ab} + \psi^a B_a + \theta^a P_a + \xi H.
\end{equation}
Its curvature $F\in\Omega^2(U,\g)$ is given by
\begin{equation}
  \label{eq:Cartan-curvature-lc}
  F = dA + \tfrac12 [A,A] = \tfrac12 \Omega^{ab} L_{ab} + \Psi^a B_a + \Theta^a P_a + \Xi H,
\end{equation}
where
\begin{equation}
  \label{eq:curvatures-lc}
  \begin{split}
    \Omega^{ab} &= d\omega^{ab} + \omega^a{}_c \wedge \omega^{cb} +
    \psi^a \wedge \theta^b - \psi^b \wedge \theta^a= F_\nabla^{ab} +
    \psi^a \wedge \theta^b - \psi^b \wedge \theta^a\\
    \Psi^a &= d\psi^a + \omega^a{}_b \wedge \psi^b + \xi \wedge \psi^a
    = d^\nabla \psi^a + \xi \wedge \psi^a\\
    \Theta^a &= d\theta^a + \omega^a{}_b \wedge \theta^b - \xi \wedge
    \theta^a = d^\nabla \theta^a - \xi \wedge \theta^a \\
    \Xi &= d\xi + \psi^a \wedge \theta_a,
  \end{split}
\end{equation}
where again $d^\nabla$ and $F_\nabla$ are the $\r$-covariant exterior
derivative and curvature, respectively.  These gauge fields satisfy
the Bianchi identity $dF + [A,F] = 0$, which decomposes into the
following relations
\begin{equation}
  \label{eq:bianchi-lc}
  \begin{split}
    d^\nabla \Omega^{ab} &= - \psi^a \wedge \Theta^b + \psi^b \wedge \Theta^a - \theta^a \wedge \Psi^b + \theta^b \wedge \Psi^a\\
    d^\nabla \Psi^a &= \Omega^a{}_b \wedge \psi^b + \psi^a \wedge \Xi - \xi \wedge \Psi^a\\
    d^\nabla \Theta^a &= \Omega^a{}_b \wedge \theta^b - \theta^a \wedge \Xi + \xi \wedge \Theta^a\\
    d \Xi &= \theta^a \wedge \Psi_a - \psi^a \wedge \Theta_a.
  \end{split}
\end{equation}
We now proceed to construct the lagrangian $4$-form out of $F$ and the
projection $\overline A$ to $\g/\h$ of $A$.  We will let $\Pbar_a$ and
$\Hbar$ denote the basis for $\g/\h$ obtained by projecting $P_a$ and
$H$, respectively, via the canonical map $\g \to \g/\h$.  Then the
projection $\overline A \in \Omega^1(U,\g/\h)$ is given by
\begin{equation}
  \overline A = \theta^a \Pbar_a + \xi \Hbar.
\end{equation}
The possible $4$-forms we can construct out of these ingredients are
$F \wedge F$, $F \wedge \overline A \wedge \overline A$ and $\overline
A \wedge \overline A \wedge \overline A \wedge \overline A$, which we
must contract with all possible $\h$-invariant tensors in the relevant
representations.  We now determine them.

\subsection{Invariant tensors}
\label{sec:invariant-tensors-lc}

The $\h$-representations of interest are $\odot^2\g^*$, $\g^* \otimes
\wedge^2(\g/\h)^*$ and $\wedge^4(\g/\h)^*$.  Let
$\lambda^{ab},\beta^a,\pi^a,\eta$ be the canonical dual basis for
$\g^*$ with $\pi^a,\eta$ the basis for $(\g/\h)^* = \h^o \subset \g^*$
canonically dual to $\Pbar_a,\Hbar$.  Since $\h$-invariant tensors are
in particular $\r$-invariant, we start by determining the possible
$\r$-invariant tensors in these representations and then look for the
subspace of $\r$-invariant tensors which are annihilated by the action
of $B_a$.  The $\r$-invariant tensors are determined using the classic
result of Weyl's \cite[Theorem~2.11.A]{MR1488158}, which says that
they are built out of $\delta_{ab}$ and $\epsilon_{abc}$.  This
results in a 14-dimensional space of $\r$-invariant tensors spanned by
the following:
\begin{equation}
  \label{eq:r-inv-tensors-lc}
  \begin{split}
    \left( \wedge^4(\g/\h)^* \right)^\r &= \left<\tfrac16 \epsilon_{abc} \pi^a \wedge \pi^b \wedge \pi^c \wedge \eta\right>\\
    \left( \odot^2\g^* \right)^\r &= \left<\tfrac14
      \lambda_{ab}\lambda^{ab}, \tfrac12 \epsilon_{abc} \lambda^{ab}
      \beta^c, \tfrac12 \epsilon_{abc} \lambda^{ab}\pi^c, \tfrac12
      \beta^a\beta_a, \tfrac12\pi^a \pi_a, \beta^a\pi_a, \tfrac12\eta^2\right>\\
    \left( \g^* \otimes \wedge^2(\g/\h)^* \right)^\r &= \left<
      \tfrac12 \lambda_{ab} \otimes \pi^a \wedge \pi^b, \tfrac12 \epsilon_{abc}
      \lambda^{ab} \otimes \pi^c \wedge \eta, \tfrac12 \epsilon_{abc} \beta^a
      \otimes \pi^b \wedge \pi^c, \right. \\
    &\qquad \left. \tfrac12 \epsilon_{abc} \pi^a \otimes \pi^b \wedge \pi^c,
    \beta_a \otimes \pi^a \wedge \eta, \pi_a \otimes \pi^a \wedge \eta  \right>.
  \end{split}
\end{equation}
To check that we have them all, we may argue as follows.  Let $V$
denote the three-dimensional vector representation of $\so(3)$ and we
shall let $\RR$ denote the one-dimensional trivial representation.
It follows that as representations of $\r\cong \so(3)$, $\g/\h \cong V
\oplus \RR$ and $\g \cong 3 V \oplus \RR$, where we have used that
$\wedge^2 V \cong V$. The same holds for their duals, so that
$(\g/\h)^* \cong V \oplus \RR$ and $\g^* \cong 3 V \oplus \RR$.  It
then follows that
\begin{equation}
   \wedge^4 (\g/\h)^*\cong \wedge^4(V \oplus \RR) \cong \RR \implies
   \dim \left( \wedge^4 (\g/\h)^*  \right)^\r = 1,
 \end{equation}
 \begin{multline}
   \odot^2 \g^*\cong \odot^2 (V \otimes \RR^3 \oplus \RR) \cong
   \wedge^2 V \otimes \wedge^2 \RR^3 \oplus \odot^2 V \otimes \odot^2
   \RR^3 \oplus \odot^2\RR\\
   \cong V \otimes \RR^3 \oplus \odot^2_0 V \otimes \RR^6 \oplus \RR \otimes \RR^6 \oplus \RR \cong 3 V \oplus 6
   \odot^2_0 V \oplus 7 \RR \implies \dim \left( \odot^2 \g^*
   \right)^\r = 7,
 \end{multline}
 and
 \begin{multline}
   \g^* \otimes \wedge^2(\g/\h)^*\cong (3 V \oplus \RR) \otimes
   \wedge^2(V \oplus \RR)\cong (3 V \oplus \RR) \otimes 2 V \cong 6 \left( V
   \otimes V\right) \oplus 2 V \cong 8 V \oplus 6 \odot^2_0 V \oplus 6 \RR\\
   \implies \dim \left(\g^* \otimes \wedge^2(\g/\h)^* \right)^\r = 6,
 \end{multline}
 where we have used that $V \otimes V = \wedge^2 V \oplus \odot^2V$
 and that $\wedge^2 V \cong V$ and $\odot^2 V = \odot^2_0 V \oplus
 \RR$, with $\odot^2_0V$ the symmetric traceless tensors.  In summary,
 we see that there space of $\r$-invariant tensors is indeed
 $14$-dimensional and hence the list in
 equation~\eqref{eq:r-inv-tensors-lc} is complete.

The action of $B_a$ on $\g/\h$ can be read off from the brackets in
\eqref{eq:lightcone-g} by simply projecting the result to $\g/\h$:
\begin{equation}
  \label{eq:b-on-g-mod-h-lc}
  \begin{split}
      B_a \cdot \Pbar_b &= \overline{[B_a,P_b]} = \delta_{ab} \Hbar\\
      B_a \cdot \Hbar &= \overline{[B_a,H]} = 0.
  \end{split}
\end{equation}
The action of $B_a$ on $\g^*$ can also be read off from the brackets
in \eqref{eq:lightcone-g}, even if not as readily, to obtain
\begin{equation}
  \label{eq:b-on-g-dual}
  \begin{split}
    B_c \cdot \lambda^{ab} &= - \lambda^{ab} \circ \ad_{B_c} = - \delta^a_c \pi^b + \delta^b_c \pi^a\\
    B_c \cdot \beta^a &= - \beta^a \circ \ad_{B_c} = \delta^a_c \eta + \delta_{cd} \lambda^{ad}\\
    B_c \cdot \pi^a &= - \pi^a \circ  \ad_{B_c} = 0\\
    B_c \cdot \eta &= - \eta \circ  \ad_{B_c} = -\delta_{ca} \pi^a,
  \end{split}
\end{equation}
where we used~\eqref{eq:dual-basis}. It is now a simple, albeit somewhat tedious, matter to use these
formulae to determine the space of $\h$-invariant tensors in the
representations of interest:
\begin{equation}
  \label{eq:inv-tensors-lc}
  \begin{split}
    \left( \wedge^4(\g/\h)^* \right)^\h &= \left<\tfrac16 \epsilon_{abc} \pi^a \wedge \pi^b \wedge \pi^c \wedge \eta\right>\\
    \left( \odot^2\g^* \right)^\h &= \left<\tfrac12 \pi^a \pi_a,
      \tfrac12 \epsilon_{abc} \lambda^{ab}\pi^c, \tfrac14
      \lambda_{ab}\lambda^{ab} - \beta_a \pi^a - \tfrac12 \eta^2\right>\\
    \left( \g^* \otimes \wedge^2(\g/\h)^* \right)^\h &= \left<\tfrac12
      \lambda_{ab} \otimes \pi^a \wedge \pi^b + \pi_a \otimes \pi^a
      \wedge \eta , \tfrac12 \epsilon_{abc} \pi^a \otimes \pi^b \wedge
      \pi^c\right>.
  \end{split}
\end{equation}
These give rise to six gauge-invariant $4$-forms listed in Table~\ref{tab:gauge-invariant-forms-lc}.

\begin{table}[h!]
  \centering
  \renewcommand{\arraystretch}{1.3}
  \begin{tabular}{>{$}l<{$}|>{$}l<{$}}\toprule
    \multicolumn{1}{c|}{invariant tensor} & \multicolumn{1}{c}{$4$-form}\\
    \midrule
    \tfrac16 \epsilon_{abc}\pi^a \wedge \pi^b \wedge \pi^c \wedge \eta  & \tfrac16 \epsilon_{abc} \theta^a \wedge \theta^b \wedge \theta^c \wedge \xi \\ 
    \tfrac12 \pi^a\pi_a & \tfrac12 \Theta^a\wedge\Theta_a\\
    \tfrac12 \epsilon_{abc} \lambda^{ab}\pi^c & \tfrac12 \epsilon_{abc} \Omega^{ab} \wedge \Theta^c\\
    \tfrac14 \lambda_{ab}\lambda^{ab} - \beta_a \pi^a - \tfrac12 \eta^2 & \tfrac14 \Omega_{ab}\wedge\Omega^{ab} - \Psi_a \wedge \Theta^a - \tfrac12 \Xi^2\\
    \tfrac12  \lambda_{ab} \otimes \pi^a \wedge \pi^b + \pi_a \otimes \pi^a \wedge \eta & \tfrac12  \Omega_{ab} \wedge\theta^a \wedge \theta^b + \Theta_a \wedge \theta^a \wedge \xi \\
    \tfrac12 \epsilon_{abc} \pi^a \otimes \pi^b \wedge \pi^c & \tfrac12 \epsilon_{abc} \Theta^a \wedge \theta^b \wedge \theta^c\\
\bottomrule
  \end{tabular}
  \caption{$\h$-invariant tensors and the corresponding gauge-invariant $4$-forms}
  \label{tab:gauge-invariant-forms-lc}
\end{table}

\subsection{Variations}
\label{sec:variations-lc}

We now proceed to vary the gauge-invariant forms in
Table~\ref{tab:gauge-invariant-forms-lc} with respect to
$\omega^{ab},\psi^a,\theta^a$ and $\xi$.  We shall discard those terms 
which upon application of the Bianchi identities \eqref{eq:bianchi-lc}
vary into exact forms.  The variation of the curvature $F$ is given as
usual by $\delta F = d\delta A + [A,\delta F]$, which unpacks into the
following variations
\begin{equation}
  \label{eq:variations-lc}
  \begin{split}
    \delta \Omega^{ab} &= d^\nabla \delta\omega^{ab} + \psi^a \wedge \delta\theta^b - \psi^b \wedge \delta\theta^a + \theta^a \wedge  \delta\psi^b - \theta^b \wedge \delta\psi^a\\
    \delta \Psi^a &= d^\nabla \delta\psi^a + \xi \wedge \delta\psi^a + \delta\omega^a{}_b \wedge \psi^b - \psi^a \wedge \delta \xi\\
    \delta \Theta^a &= d^\nabla \delta\theta^a - \xi \wedge \delta\theta^a  + \delta\omega^a{}_b \wedge \theta^b + \theta^a\wedge \delta \xi\\
    \delta \Xi &= d\delta\xi + \psi_a \wedge \delta\theta^a - \theta_a \wedge \delta\psi^a.
  \end{split}
\end{equation}
Varying the gauge-invariant $4$-forms in
Table~\ref{tab:gauge-invariant-forms-lc} we find that $\tfrac14
\Omega_{ab}\wedge\Omega^{ab} - \Psi_a \wedge \Theta^a - \tfrac12
\Xi^2$ varies into an exact form:
\begin{equation}
  \delta\left(\tfrac14 \Omega_{ab}\wedge\Omega^{ab} - \Psi_a \wedge
  \Theta^a - \tfrac12 \Xi^2\right) = d\left( \tfrac12
  \Omega_{ab}\wedge \delta\omega^{ab} - \Xi \wedge \delta\xi -
  \Theta_a \wedge \delta\psi^a - \Psi_a \wedge \delta\theta^a\right).
\end{equation}
Similarly, the following two linear combinations also vary into exact
forms:
\begin{multline}
  \delta\left(\tfrac12  \Omega_{ab} \wedge\theta^a \wedge \theta^b +
    \Theta_a \wedge \theta^a \wedge \xi - \tfrac12 \Theta_a \wedge
    \Theta^a  \right)\\
  = d \left( \tfrac12 \delta\omega_{ab} \wedge \theta^a \wedge
    \theta^b + \delta\theta^a \wedge \theta_a \wedge \xi - \Theta_a
    \wedge \delta\theta^a\right)
\end{multline}
and
\begin{equation}
  \delta\left(\tfrac12 \epsilon_{abc} \Theta^a \wedge \theta^b \wedge
    \theta^c- \tfrac12 \epsilon_{abc} \theta^a \wedge \theta^b \wedge
    \theta^c \wedge \xi  \right) = d \left( \tfrac12 \epsilon_{abc}
    \theta^a \wedge \theta^b \wedge \delta\theta^c \right).
\end{equation}

This leaves the following gauge-invariant $4$-form lagrangian
\begin{equation}
  \label{eq:lagrangian-lc}
  \eL = \tfrac\beta2 \Theta^a \wedge \Theta^a + \tfrac\mu2
  \epsilon_{abc}\Omega^{ab} \wedge \Theta^c + \tfrac\Lambda6
  \epsilon_{abc} \theta^a \wedge \theta^b \wedge \theta^c \wedge \xi.
\end{equation}

\subsection{Euler--Lagrange equations}
\label{sec:euler-lagr-equat-lc}

Varying the lagrangian $4$-form in \eqref{eq:lagrangian-lc} we obtain
\begin{equation}
  \delta \eL = d \Upsilon + \tfrac12 D_{ab} \wedge \delta\omega^{ab} + E_a
  \wedge \delta\psi^a + F_a \wedge \delta\theta^a + G \wedge \delta \xi,
\end{equation}
where
\begin{equation}
  \Upsilon = \beta \Theta_a \wedge \delta\theta^a + \tfrac\mu2
  \epsilon_{abc} \left( \Theta^a \wedge \delta\omega^{bc} +
    \Omega^{ab} \wedge \delta\theta^c \right)
\end{equation}
and the Euler--Lagrange equations are given by the vanishing of the following:
\begin{equation}
  \label{eq:euler-lagr-eqns-lc}
  \begin{split}
  D_{ab} &= \beta \left( \Theta_b \wedge \theta_a - \Theta_a \wedge
    \theta_b \right) - \mu \epsilon_{abc} \left( \Omega^c{}_d \wedge \theta^d - \theta^c \wedge \Xi  + \xi \wedge \Theta^c\right) - \tfrac\mu2 \Omega^{cd} \wedge \left( \epsilon_{acd} \theta_b - \epsilon_{bcd} \theta_a\right)\\
  E_a &= - \mu \epsilon_{abc} \theta^b \wedge \Theta^c\\
  F_a &= \beta \left(\theta_a \wedge \Xi - \Omega_{ab} \wedge \theta^b - 2 \xi \wedge \Theta_a \right)  + \mu \epsilon_{abc} \left( \theta^b \wedge \Psi^c - \tfrac12 \Omega^{bc} \wedge \xi \right) - \tfrac\Lambda2 \epsilon_{abc} \xi \wedge \theta^b \wedge \theta^c\\
  G &= \beta \Theta_a \wedge \theta^a + \tfrac\mu2 \epsilon_{abc} \Omega^{ab} \wedge \theta^c + \tfrac\Lambda6 \epsilon_{abc}\theta^a \wedge \theta^b \wedge \theta^c.
\end{split}
\end{equation}

If $\beta = \mu = 0$, then the $F_a$ and $G$ equations have no
solution unless $\Lambda = 0$ as well in which case the lagrangian is
identically zero.  Hence at least one of $\beta,\mu$ must be
nonzero.  At first let us assume that $\mu \neq 0$.  It will turn
out that the equations do not depend on $\beta$.

The $E_a$ equation says that $\theta^{[a} \wedge \Theta^{b]} = 0$.
Inserting that into the $D_{ab}$ equation we have that
\begin{equation}
  \epsilon_{abc} \left( \Omega^c{}_d \wedge \theta^d - \theta^c \wedge
    \Xi  + \xi \wedge \Theta^c\right) + \tfrac12 \Omega^{cd} \wedge
  \left( \epsilon_{acd} \theta_b - \epsilon_{bcd} \theta_a\right) = 0.
\end{equation}
Since this is skewsymmetric in $ab$ we can contract with
$\epsilon^{abe}$ to arrive at the equivalent equation $\xi \wedge
\Theta^a = \theta^a \wedge \Xi$.  In other words, we have arrived at
equations \eqref{eq:torsion-free-1} and \eqref{eq:torsion-free-2} and
hence by Lemma~\ref{lem:torsion-free-adsc}, the torsion vanishes:
$\Theta^a = 0$ and $\Xi = 0$.  Their Bianchi identities then say that
\begin{equation}\label{eq:torsion-free-bianchi-lc}
  \theta^a
  \wedge \Psi_a = 0 \qquad\text{and}\qquad \Omega^a{}_b \wedge \theta^b = 0.
\end{equation}
The $F_a$ equation then becomes
\begin{equation}
  \mu \epsilon_{abc} \left( \theta^b \wedge \Psi^c - \tfrac12
    \Omega^{bc} \wedge \xi \right) - \tfrac\Lambda2 \epsilon_{abc} \xi
  \wedge \theta^b \wedge \theta^c = 0,
\end{equation}
and in turn the $G$ equation becomes
\begin{equation}
  \tfrac\mu2 \epsilon_{abc} \Omega^{ab} \wedge \theta^c +
  \tfrac\Lambda6 \epsilon_{abc}\theta^a \wedge \theta^b \wedge
  \theta^c = 0.
\end{equation}
To go further let us expand the curvatures $\Omega^{ab}$ and $\Psi^a$
in terms of the vielbeins:
\begin{equation}
  \begin{split}
    \Omega_{ab} &= \tfrac12 R_{abcd} \theta^c \wedge \theta^d + S_{abc} \theta^c \wedge \xi\\
    \Psi_a &= \tfrac12 T_{abc} \theta^b \wedge \theta^c + U_{ab} \theta^b \wedge \xi.
  \end{split}
\end{equation}
We shall also introduce the shorthand $\lambda := \frac\Lambda\mu$.
The $G$ equation says that
\begin{equation}
  \tfrac12 \epsilon_{abc} \left( \tfrac12 R^{ab}{}_{de} \theta^d \wedge \theta^e + S^{ab}{}_d \theta^d \wedge \xi + \tfrac13 \lambda \theta^a \wedge \theta^b\right) \wedge \theta^c = 0,
\end{equation}
which breaks up into two equations:
\begin{equation}
  \begin{split}
    \tfrac12 \epsilon_{abc} S^{ab}{}_d \theta^c \wedge \theta^d \wedge \xi = 0 & \implies S^{ab}{}_a = 0\\
    \tfrac14 \epsilon_{abc} R^{ab}{}_{de} \theta^c \wedge \theta^d \wedge \theta^e + \tfrac16 \lambda \epsilon_{abc}\theta^a \wedge \theta^b \wedge \theta^c = 0 & \implies R = 2 \lambda,
  \end{split}
\end{equation}
where $R := R^{ab}{}_{ba}$.  The $F_a$ equation becomes
\begin{equation}
 \epsilon_{abc} \theta^b \wedge \left(  \tfrac12 T^c{}_{de}\theta^d \wedge \theta^e + U^c{}_d \theta^d \wedge \xi\right) - \tfrac14 \epsilon_{abc} R^{bc}{}_{de} \theta^d \wedge \theta^e \wedge \xi- \tfrac12 \epsilon_{abc} \lambda \theta^b \wedge \theta^c \wedge \xi = 0,
\end{equation}
which again breaks up into two equations.  The first equation says
\begin{equation}
  \tfrac12 \epsilon_{abc} T^c{}_{de} \theta^b \wedge \theta^d \wedge \theta^e = 0 \implies T^b{}_{ba} = 0,
\end{equation}
and the second equation says that
\begin{equation}\label{eq:Fa-eqn-lc}
  \left( \epsilon_{adc} U^c{}_e - \tfrac14 \epsilon_{abc} R^{bc}{}_{de} - \tfrac12 \epsilon_{ade} \lambda  \right) \theta^d \wedge \theta^e \wedge \xi = 0,
\end{equation}
which we will analyse after we take into account the Bianchi
identities~\eqref{eq:torsion-free-bianchi-lc}.  The first of the
Bianchi identities becomes
\begin{equation}
  \theta^a \wedge \left( \tfrac12 T_{abc} \theta^b \wedge \theta^c +
    U_{ab} \wedge \theta^b \wedge \xi \right) = 0 \implies T_{[abc]} =
  0 \quad\text{and}\quad U_{[ab]} = 0.
\end{equation}
The second Bianchi identity becomes
\begin{equation}
  \left( \tfrac12 R_{abcd} \theta^c \wedge \theta^d + S_{abc} \theta^c
  \wedge \xi\right) \wedge \theta^b = 0 \implies R_{a[bcd]} = 0
\quad\text{}\quad S_{a[bc]} = 0.
\end{equation}
The latter condition, together with the fact that $S_{abc} = - S_{bac}$
says that $S_{abc} = 0$; whereas the former condition says that
$R_{abcd}$ is an algebraic curvature tensor, which in three dimensions
is determined by its Ricci tensor $R_{ad} := \delta^{bc} R_{abcd}$ via
the Ricci decomposition:
\begin{equation}
\label{eq:ricci-decomp-lightcone}
  R_{abcd} = R_{ad} \delta_{bc} - R_{bd} \delta_{ac} + R_{bc}
  \delta_{ad} - R_{ac} \delta_{bd} + \tfrac12 R (\delta_{ac}
  \delta_{bd} - \delta_{bc}\delta_{ad}).
\end{equation}
We now return to equation~\eqref{eq:Fa-eqn-lc}, which after relabelling
indices and using the definition of the Ricci tensor and the fact that
$R = 2 \lambda$, becomes
\begin{equation}
  U_{ab} - U \delta_{ab} + R_{ab} = 0,
\end{equation}
where $U = \delta^{ab}U_{ab}$.  Taking the trace we see that $U = 
\tfrac12 R$ and hence that $-U_{ab}$ is the Einstein tensor:
\begin{equation}
  U_{ab} = -(R_{ab} - \tfrac12 R \delta_{ab}) = -R_{ab} + \lambda \delta_{ab}.
\end{equation}
Finally, let us write $T^a{}_{bc} = \epsilon_{bcd} W^{ad}$.  Then
$T^a{}_{ab}= 0$ and $T_{[abc]} = 0$ imply that $W^{ab}$ is symmetric
and traceless.

We may summarise the preceding discussion as follows, where we have
put $\mu = 1$ without loss of generality.

\begin{proposition}\label{prop:solution-lc}
  The solution of the Euler--Lagrange equations for the lagrangian
  \begin{equation*}
    \eL =
     \tfrac12  \epsilon_{abc}\Omega^{ab} \wedge \Theta^c
    + \tfrac\beta2 \Theta^a \wedge \Theta_a
    + \tfrac\Lambda6 \epsilon_{abc} \theta^a \wedge \theta^b \wedge \theta^c \wedge \xi
  \end{equation*}
  are such that $\Xi = \Theta^a = 0$ and
  \begin{equation}
  \label{eq:solution-for-lc}
    \begin{split}
      \Omega_{ab} &= R_{ac} \theta_b \wedge \theta^c - R_{bc} \theta_a
      \wedge \theta^c + 2\Lambda \theta_a \wedge \theta_b\\
      \Psi_a &= \tfrac12 \epsilon^{bcd}  W_{ab} \theta_c \wedge
      \theta_d - \left( R_{ab} - \Lambda \delta_{ab}\right) \theta^b \wedge \xi,
      \end{split}
  \end{equation}
  where $W_{ab}$ is symmetric and traceless.
\end{proposition}

\section{Conclusion and outlook}
\label{sec:conclusion}

In this work we framed the construction of gravitational lagrangians
via the ``gauging of spacetime algebras'' in terms of Cartan geometry.
We emphasized that, in this procedure, it is in fact not a spacetime
symmetry algebra that is being gauged but rather a Klein pair
$(\g,\h)$ determining the underlying homogeneous space, or flat model,
of the Cartan geometry. In Section \ref{sec:gauging-poincare} we
reviewed the classic example of the Klein pair
$(\mathfrak{iso}(3,1),\mathfrak{so}(3,1))$ corresponding to
four-dimensional Minkowski spacetime. The resulting lagrangian
\eqref{eq:relativistic-lagrangian} turned out to be the
Hilbert-Palatini lagrangian together with the Holst term, a
cosmological constant term and topological terms (Nieh--Yan,
Pontryagin and Gauss--Bonnet). In Sections \ref{sec:gauging-AdSC} and
\ref{sec:carroll-gauging} we subsequently applied the gauging
procedure to the carrollian homogeneous spaces $\zAdSCp$ and $\zC$
spacetime, which resulted in the lagrangians \eqref{eq:lagrangian} and
\eqref{eq:Carroll-lagrangian}, respectively.  The terms in these
lagrangians can be identified with carrollian counterparts of
Hilbert--Palatini, Holst, and cosmological constant terms.  In
addition, there are carrollian analogues of the Nieh--Yan, Pontryagin
and Gauss--Bonnet topological terms, as well as other topological
terms (in the reductive cases) which are intrinsically carrollian and
seem to have no relativistic analogues.  Furthermore, we showed
explicitly in Section~\ref{sec:carroll-gauging} that the lagrangian of
flat Carroll space can be recovered as a limit of a relativistic
lagrangian.  Finally, we applied the gauging procedure to the
four-dimensional lightcone, the only remaining four-dimensional
carrollian space in the classification of
\cite{Figueroa-OFarrill:2018ilb}, and determined its lagrangian to be
given by \eqref{eq:lagrangian-lc}.

We will close this work with a few comments on the generality of the
approach and possible applications of the theories constructed herein.

\textbf{Do we get all gauge-invariant lagrangians?} The approach
presented in this work is rather general, as it can be applied to the
construction of gravitational theories for all kinematical spacetimes
considered, e.g., in \cite{Figueroa-OFarrill:2018ilb}. However, given
a kinematical spacetime one cannot recover all gauge-invariant,
gravitational lagrangians with second-order equations of motion using
the gauging procedure. Consider, for instance, the lagrangian for
``electric Carroll gravity''
\begin{equation}
  \label{eq:electriccarroll}
  \eL_{e} = \tfrac14d\text{vol} (\mathcal{L}_v h)_{\mu\nu}(\mathcal{L}_v h)_{\rho\sigma}\left(\gamma^{\mu\rho}\gamma^{\nu\sigma} - \gamma^{\mu\nu}\gamma^{\rho\sigma} \right)
\end{equation}
originally found by Henneaux in~\cite{Henneaux:1979vn} that recently
resurfaced in works such as~\cite{Henneaux:2021yzg,Hansen:2021fxi}.
This lagrangian leads to second-order equations of motion and can be
obtained from a limiting procedure of Einstein gravity in the
second-order formulation using either
hamiltonian~\cite{Henneaux:1979vn,Henneaux:2021yzg} or
lagrangian~\cite{Hansen:2021fxi} methods.  However, as we will discuss
below, it is not equivalent to the lagrangian \eqref{eq:lagrangian}
for any choice of parameters.  It thus seems likely that also for the
other carrollian Klein pairs there exist theories that are not
described by the lagrangians we have constructed in this paper.

\textbf{A different route to building invariant lagrangians.} With the
geometric interpretation of Section~\ref{sec:geom-interpr} comes a
different way\footnote{In contrast to Section~\ref{sec:geom-interpr},
  we no longer take the torsion to be zero. This means that the affine
  connection is no longer given
  by~\eqref{eq:torsion-free-adapted-connection}, although the
  expression~\eqref{eq:affine-conn}, relating the affine connection to
  the Cartan connection continues to hold.} of building $H$-gauge
invariant lagrangians, namely by using that any gauge-invariant
$4$-form $X$ is of the form
\begin{equation}
  \label{eq:gaugeX}
  X = F\,d\text{vol}
\end{equation}
where $F\in C^\infty(M)$ is a locally $H$-invariant function built
from the carrollian structure and its ``inverse''. This perspective
was taken in~\cite{Hartong:2015xda} and allows for the construction of
a much larger class of lagrangians.

Given a carrollian structure $(\kappa^\mu,h_{\mu\nu})$, the simplest
invariant object that we can write down is the second fundamental form
\begin{equation}
\label{eq:intrinsic-torsion}
    K_{\mu\nu} = \tfrac12 \mathcal{L}_\kappa h_{\mu\nu},
\end{equation}
which captures the intrinsic torsion of the carrollian
structure~\cite{Figueroa-OFarrill:2020gpr}. By construction, this object
satisfies $\kappa^\mu K_{\mu\nu}$, and is thus \textit{spatial}. Using
Cartan's magic formula, we may relate $K_{\mu\nu}$ to the torsion
$\Theta^a$, since
\begin{equation}
    \mathcal{L}_\kappa\theta^a = \iota_\kappa \Theta^a - \iota_\kappa(\omega^a{_b})\theta^b,
\end{equation}
which in components becomes
\begin{equation}
    \mathcal{L}_\kappa \theta_\mu{^a} = \kappa^\rho\Theta_{\rho\mu}{^a} - \kappa^\rho \omega_\rho{^a}{_b}\theta_\mu{^b}.
\end{equation}
This means that
\begin{equation}
  \label{eq:Ktorsion}
    K_{\mu\nu} = 2\delta_{ab} \theta^a_{(\mu}\mathcal{L}_\kappa\theta_{\nu)}{^b} = 2\kappa^\rho \Theta_{\rho(\mu}\theta_{\nu)a} - 2\kappa^\rho \omega_{\rho ab}\theta_{(\mu}{^a}\theta_{\nu)}{^b} = 2\kappa^\rho \Theta_{\rho(\mu}{^a}\theta_{\nu)a},
\end{equation}
where we used antisymmetry of $\omega_{ab}$. Now, using the object
$\gamma^{\mu\nu} = \delta_{ab}e^\mu{_a}e^\nu{_b}$, which transforms as
follows under local Carroll boosts
\begin{equation}
    \delta \gamma^{\mu\nu} = 2\kappa^{(\mu}\lambda^{\nu)},
\end{equation}
where $\lambda^\mu = \gamma^{\mu\nu}\lambda_\nu$ is the boost
parameter, we may build the following invariants out of the second
fundamental form~\eqref{eq:intrinsic-torsion}
\begin{align}
    \gamma^{\mu\rho}\gamma^{\nu\sigma}K_{\mu\nu}K_{\rho\sigma} &:= K^{\mu\nu}K_{\mu\nu} & \gamma^{\mu\nu}K_{\mu\nu} &:= K
\end{align}
Using these objects, the simplest Carroll-invariant lagrangian that is
of the second order in derivatives takes the form
\begin{equation}
  \label{eq:electriccarrollK}
  \eL = d\text{vol}\,\left(K^{\mu\nu}K_{\mu\nu} - K^2 \right).
\end{equation}
which is the electric theory \eqref{eq:electriccarroll}.

The preceeding discussion also makes clear the fact that the electric
theory is not equivalent to any of the theories we constructed for any
choice of parameters. The theory \eqref{eq:electriccarrollK} allows
for solutions with non-vanishing $K_{\mu\nu}$, see e.g.,
\cite{Hansen:2021fxi} for the explicit equations of motion. However,
for all theories we construct the component $\kappa \cdot \Theta^{a}$
of the torsion is set to zero, even though $\Theta^a$ does not need to
vanish in general, e.g., for the lagrangian \eqref{eq:lagrangian} with
$\mu=0$. Given equation \eqref{eq:Ktorsion} it is then clear that for
all of our theories $K_{\mu\nu}=0$.

It would be interesting to see whether at least some of the more
general theories of the form \eqref{eq:gaugeX}, in particular the
electric theory \eqref{eq:electriccarrollK}, can be obtained by
coupling the lagrangians considered in this work to suitably chosen
matter fields or by loosening some of the assumptions of our gauging
procedure.  We leave this for future study.

\textbf{Matter couplings.}  In this paper we have restricted attention
to pure gravity theories, but as in any gauge theory, gravity theories
may be coupled to matter.  A Cartan connection on a principal
$H$-bundle $P \to M$ defines a covariant derivative on sections of any
associated bundle.  For example, if $V$ is a linear representation of
$H$, we may form an associated vector bundle $P \times_H V$, whose
sections can be interpreted as $H$-equivariant functions $P \to V$.
The covariant derivative induced by the Cartan connection on such
sections allows us to write down gauge-invariant terms in the
lagrangian describing the coupling of such matter to gravity.
Similarly, if $N$ is a manifold on which $H$ acts smoothly, we may
form an associated fibre bundle $P \times_H N$ whose sections carry a
nonlinear realisation of $H$ and which may be differentiated
covariantly.  It might be interesting to explore Cartan geometries
defined via the Euler--Lagrange equations involving matter couplings.

\textbf{Topological terms and asymptotic symmetries.} In the case of
Einstein gravity in asymptotically flat spacetimes, one finds that the
Poincaré symmetries are enhanced asymptotically to
infinite-dimensional BMS symmetries \cite{Bondi:1962px,Sachs:1962zza},
which have recently garnered attention in the context of flat space
holography. It has been argued that these symmetries need to be further
enhanced by so-called dual supertranslations that arise most naturally
from the first-order lagrangian \eqref{eq:relativistic-lagrangian}
with the Holst term included \cite{Godazgar:2020gqd,Godazgar:2020kqd}.
Asymptotic symmetries of carrollian gravity were studied in
\cite{Perez:2021abf,Perez:2022jpr}. It would be interesting to see
whether the symmetries found in these works get additional
contributions from the analogous terms in our carrollian lagrangians
\eqref{eq:lagrangian} and \eqref{eq:Carroll-lagrangian}.

\textbf{Other dimensions.} We have considered only $(3+1)$-dimensional
theories. However, it is clear that our approach is applicable to any
spacetime dimension depending on which one could find additional
contributions to the lagrangian. For instance, for odd-dimensional
spacetimes one could add Chern--Simons terms to the lagrangian. For
lower-dimensional models of carrollian gravity see, e.g.,
\cite{Bergshoeff:2016soe,Matulich:2019cdo,Ravera:2019ize} for $(2+1)$
and \cite{Grumiller:2020elf,Gomis:2020wxp} for $(1+1)$-dimensional
toy-models. In a different context, the $(2+1)$-dimensional lightcone
theory has been discussed in \cite{Nguyen:2020hot}.

\textbf{Solution spaces.} In this work we have largely restricted
ourselves to the construction of gravitational lagrangians for
carrollian theories and analysing the equations of motion. We have not
attempted to obtain solutions to the corresponding equations of motion
in terms of the vielbeins. As is well-known, the equations of motion of
the relativistic lagrangian \eqref{eq:relativistic-lagrangian} lead to
a plethora of interesting metrics such as gravitational waves or
(colliding) black holes to name but two such classes of metrics.
Thinking of the carrollian theories as arising from a limit of
General Relativity, one would expect the solution spaces of the
carrollian theories to contain fewer interesting solutions, which,
however, might be easier to construct.  Nevertheless, these solutions
can still have interesting physical interpretations, such as
describing the dynamics of General Relativity near a spatial
singularity; cf.~\cite{Henneaux:1979vn}, for example.  A better
understanding of the equations of motion of the carrollian theories
would therefore be of interest, cf.~\cite{Perez:2021abf,
  Guerrieri:2021cdz, deBoer:2021jej, Hansen:2021fxi, Perez:2022jpr}.

\textbf{Gravitational vacua and (pseudo-)carrollian spaces.} The
above-mentioned enhancement of Poincaré to BMS at null infinity in
asymptotically flat spacetimes implies that the (radiative) vacuum of gravity in
asymptotically flat spacetimes is infinitely degenerate
\cite{Ashtekar:1981hw}. The different gravitational vacua are related
by supertranslations and superrotations. At null infinity, these vacua
can be understood in terms of Cartan geometry based on a certain
homogeneous space of the Poincaré group
\cite{Herfray:2021xyp,Herfray:2021qmp}. In \cite{Nguyen:2020hot}, it
was shown that the superrotation sector can be derived from a
three-dimensional action.

The enhancement of Poincaré to BMS can also be demonstrated at spatial
\cite{Troessaert:2017jcm,Henneaux:2018cst,Henneaux:2019yax} and
time-like infinity \cite{Chakraborty:2021sbc}. In the recent work
\cite{Figueroa-OFarrill:2021sxz}, we showed that both space-like and
time-like infinities can be understood as (pseudo-)carrollian,
homogeneous spaces of the Poincaré group. In particular, the blow-up
of time-like infinity can be described by $\zAdSC$. It is thus
suggestive that (a constrained version of) the lagrangian
\eqref{eq:lagrangian} for $\zAdSC$ describes the BMS vacuum sector of
General Relativity near time-like infinity. Similar comments apply to
the case of space-like infinity with associated pseudo-carrollian
Klein pair $\zSpi$, that is related to the Ashtekar--Hansen structure
at spatial infinity \cite{Ashtekar:1978zz,Gibbons:2019zfs}. While we
have not discussed the gauging of this Klein pair, it can be obtained
from \eqref{eq:lagrangian} with only minor adjustments.

\textbf{General Relativity: The view from timelike infinity.} More
speculative, but also more rewarding, is the idea that the relation
between General Relativity on asymptotically flat spacetimes and the
gaugings of $\zAdSC$ and $\zSpi$ extends also away from the vacuum
sector. In \cite{Figueroa-OFarrill:2021sxz} it was shown how points in
Minkowski spacetime correspond to certain higher-dimensional geometries in
$\zAdSC/\zSpi$. Perhaps the lagrangian \eqref{eq:lagrangian} based on
the Klein pair of $\zAdSC$ and its solutions encode certain aspects of
General Relativity on asymptotically flat spacetimes. Since the map
between the flat models of the two Cartan geometries, Minkowski and
$\zAdSC$/$\zSpi$, is non-local, the putative construction would likely
have a twistorial flair. We leave this interesting possibility for
further study.

\section*{Acknowledgments}
\label{sec:acknowledgments}

We are grateful to Jelle Hartong, Yannick Herfray, Niels Obers and
Alfredo Pérez for useful discussions.

The work of EH is supported by the Royal Society Research Grant for
Research Fellows 2017 “A Universal Theory for Fluid Dynamics” (grant
number RGF$\backslash$R1$\backslash$180017).

SP was supported by the Leverhulme Trust Research Project Grant
(RPG-2019-218) ``What is Non-Relativistic Quantum Gravity and is it
Holographic?''.

JS was supported by a Marina Solvay-fellowship and the F.R.S.-FNRS
Belgium through the convention IISN 4.4503.

SP and JS acknowledge support of the Erwin Schrödinger Institute (ESI)
in Vienna where part of this work was conducted during the thematic
programme ``Geometry for Higher Spin Gravity: Conformal Structures,
PDEs, and Q-manifolds''. EH and SP would like to thank the organisers
of the Carroll Workshop in Vienna where part of this work was completed.

\appendix

\section{The Ricci tensor of a torsion-free affine connection}
\label{sec:ricci-tensor-torsion}

In this appendix we give a proof of the following well-known result.

\begin{proposition}\label{prop:symmetric-ricci}
  Let $M$ be an $n$-dimensional orientable manifold with a torsion-free affine
  connection $\nabla$.  Then the Ricci tensor $r^\nabla$ is symmetric
  if and only if around every point of $M$ there exists a locally
  defined parallel volume form.
\end{proposition}

\begin{proof}
  Let $\omega \in \Omega^n(M)$ be an orientation.  Then for all vector
  fields $X \in \eX(M)$, the $n$-form $\nabla_X \omega$ is
  proportional to $\omega$: $\nabla_X \omega = f_X \omega $ for some
  function $f_X \in C^\infty(M)$ depending on $X$. Since $\nabla_X$ is
  tensorial in $X$, there exists a one-form $\alpha \in \Omega^1(M)$
  such that $f_X = \alpha(X)$.  The action of the curvature operator
  $R^\nabla(X,Y)$ on $\omega$ is
  \begin{equation}
    \begin{split}
      R^\nabla(X,Y) \cdot \omega &= \nabla_{[X,Y]}\omega - [\nabla_X,\nabla_Y]\omega\\
      &= \alpha([X,Y]) \omega - \nabla_X (\alpha(Y)\omega) + \nabla_Y (\alpha(X)\omega)\\
      &= \alpha([X,Y]) \omega - (\nabla_X\alpha)(Y)\omega -
      \alpha(\nabla_XY)\omega - \alpha(Y)\nabla_X \omega\\
      & \qquad {} + (\nabla_Y\alpha)(X)\omega +
      \alpha(\nabla_YX)\omega + \alpha(X)\nabla_Y \omega\\
      &= - (\nabla_X\alpha)(Y) \omega + (\nabla_Y\alpha)(X) \omega\\
      &= - d\alpha(X,Y) \omega,
    \end{split}
  \end{equation}
  where we have used that $\nabla$ is torsion-free.  Now the curvature
  operator $R^\nabla(X,Y)$ acts on $n$-forms via scalar multiplication
  by the negative of the trace $\Tr (Z \mapsto R^\nabla(X,Y)Z)$ of the
  curvature operator acting on vector fields.  Using the algebraic
  Bianchi identity (in the absence of torsion) and the linearity of
  the trace, we have that
  \begin{equation}
    \begin{split}
      -\Tr (Z \mapsto R^\nabla(X,Y)Z) &= \Tr(Z \mapsto R^\nabla(Z,X)Y) +
      \Tr(Z \mapsto R^\nabla(Y,Z)X)\\
      &= \Tr(Z \mapsto R^\nabla(Z,X)Y) - \Tr(Z \mapsto R^\nabla(Z,Y)X)\\
      &=-r^\nabla(X,Y) + r^\nabla(Y,X).
    \end{split}
  \end{equation}
  In other words, we have that
  \begin{equation}
    r^\nabla(X,Y) - r^\nabla(Y,X) = d\alpha(X,Y),
  \end{equation}
  so that the Ricci tensor is symmetric if and only if $d\alpha=0$.  By
  the Poincaré Lemma, $d\alpha = 0$ if and only if locally $\alpha = df$
  for some locally-defined function $f$.  But in this case,
  \begin{equation}
    \nabla_X \omega = df(X)\omega = X(f)\omega
  \end{equation}
  and we can then modify $\omega$ locally to $\vol := e^{-f}\omega$ so
  that $\nabla_X \vol = 0$.
\end{proof}

In Section~\ref{sec:geom-interpr} we showed that the Cartan connection
defines a locally-defined volume form which is parallel relative to
the affine connection $\nabla$ it defines.  Therefore, we conclude
that its Ricci tensor is symmetric.

Let $\nabla$ be a torsion-free affine connection on an orientable
manifold $M$ whose Ricci tensor is symmetric and let
$\widetilde\nabla$ be a second torsion-free affine connection on $M$.
Let $K_X Y := \nabla_X Y - \widetilde\nabla_X Y$ denote the contorsion
$(1,2)$-tensor field and define a one-form $\beta \in \Omega^1(M)$ by
$\beta(X) := \Tr (Y \mapsto K_X Y)$.

\begin{proposition}\label{prop:contorsion-ricci-symmetric}
  The Ricci tensor of $\widetilde\nabla$ is symmetric if and only if
  $d\beta=0$.
\end{proposition}

\begin{proof}
  By Proposition~\ref{prop:symmetric-ricci}, the Ricci tensor
  $r^{\widetilde\nabla}$ is symmetric if and only if there exists
  about every point a locally-defined volume form $\widetilde\vol$ which is
  parallel with respect to $\widetilde\nabla$.  Let $\vol$ be a
  locally-defined parallel volume form for $\nabla$.  Then
  $\widetilde \vol = f \vol$ for some locally-defined
  nowhere-vanishing function $f$.  Without loss of generality we can
  assume that $\vol$ and $\widetilde\vol$ define the same orientation
  so that $f$ is positive, say $f = e^g$ for some locally-defined
  smooth function $g$.  Then
  \begin{equation}
    \begin{split}
      \widetilde\nabla_X \widetilde\vol &= \widetilde\nabla_X (e^g \vol)\\
      &= e^g X(g) \vol + e^g \widetilde\nabla_X \vol\\
      &= X(g) e^g \vol + e^g K_X\cdot \vol\\
      &= e^g \left( X(g) - \beta(X) \right) \vol,
    \end{split}
  \end{equation}
  so that $\widetilde\vol$ is parallel if and only if $\beta(X) = X(g)
  = dg(X)$ for all vector fields $X$ and for some function $g$.  In
  other words, $\beta = dg$, so that it is locally exact or,
  equivalently, it is closed.
\end{proof}

\providecommand{\href}[2]{#2}\begingroup\raggedright\endgroup

\end{document}